\documentclass[envcountsame,envcountsect,runningheads,a4paper]{llncs}
\usepackage{wrapfig}
\usepackage{listings}
\usepackage[dvipdfmx]{graphicx}
\usepackage{tikz}
\usetikzlibrary{automata,positioning}
\usepackage[numbers,sectionbib,sort]{natbib}
\usepackage{amsmath,amssymb,stmaryrd}
\usepackage{alltt}
\usepackage{bcpproof}
\usepackage[dvipdfmx]{hyperref}

\newif\iffull\fulltrue
\newif\ifdraft\draftfalse
\newif\ifconffinal\conffinalfalse

\usepackage[linesnumbered,lined]{algorithm2e}
\newcommand\tuple[1]{\left\langle {#1} \right\rangle}
\newcommand\states{Q}
\newcommand\st{q}
\newcommand\vars{V}
\newcommand\var{x }

\newcommand\flow{F}
\newcommand\inv{\mathit{inv}}
\newcommand\predtrans{\mathcal{F}}
\newcommand\predtranscont{\mathcal{F_C}}
\newcommand\predtranshybrid{\mathcal{F_H}}

\newcommand\SYSTEM{\mathcal{S}}
\newcommand\dtsts{\SYSTEM_D}
\newcommand\hsts{\SYSTEM_H}

\newcommand\set[1]{\left\{{#1}\right\}}
\newcommand\REAL{\mathbb{R}}
\newcommand\todo[1]{\ifdraft\textbf{TODO: {#1}}\else\fi}

\newcommand\ra{\rightarrow}
\newcommand\la{\leftarrow}
\newcommand\Lra{\leadsto}

\newcommand\VALUATIONS{\Sigma}
\newcommand\valuation{\sigma}
\DeclareMathOperator{\COL}{{:}}
\DeclareMathOperator{\DEFEQ}{{:=}}

\newcommand\fml{\varphi}
\newcommand\setfml{\mathbf{Fml}}

\newcommand\trans{\delta}

\newcommand\ODE{\mathcal{D}}
\newcommand\sem[1]{\left\llbracket{#1}\right\rrbracket}

\newcommand\INITIALIZE{\textsc{Initialize}}
\newcommand\VALID{\textsc{Valid}}
\newcommand\UNFOLD{\textsc{Unfold}}
\newcommand\INDUCTION{\textsc{Induction}}
\newcommand\CANDIDATE{\textsc{Candidate}}
\newcommand\DECIDE{\textsc{Decide}}

\newcommand\MODEL{\textsc{Model}}

\newcommand\CONFLICT{\textsc{Conflict}}
\newcommand\INDUCTIONCONT{\textsc{PropagateCont}}
\newcommand\CANDIDATECONT{\textsc{CandidateCont}}
\newcommand\DECIDECONT{\textsc{DecideCont}}
\newcommand\CONFLICTCONT{\textsc{ConflictCont}}

\newcommand\RESVALID{\mathbf{Valid}}
\DeclareMathOperator\RESMODEL{\mathbf{Model}}
\newcommand\RESCONT{\mathbf{Cont}}

\newcommand\PDRState[2]{#1\> || \> #2}
\newcommand\cetrace{M}
\newcommand\abstraction{A}

\newcommand\rel{R}

\newcommand\cmd{\varphi_c}
\newcommand\SKIP{\mathbf{skip}}

\DeclareMathOperator\IF{\mathbf{if}}
\DeclareMathOperator\THEN{\mathbf{then}}
\DeclareMathOperator\ELSE{\mathbf{else}}
\DeclareMathOperator\AND{\mathbf{and}}

\newcommand\TRUE{\mathit{true}}
\newcommand\FALSE{\mathit{false}}

\newcommand\HYBRIDPDR{\textsc{HGPDR}}
\newcommand\DETHYBRIDPDR{\textsc{DetHybridPDR}}

\newcommand\CONTIREACH[2]{\operatorname{\rightarrow}_{{#1},{#2}}}

\newcommand\CONTIREACHPRED[2]{\langle{#1} \mid {#2}\rangle}

\newcommand\rem{\mathit{rem}}

\newcommand\RED[1]{\ra_{#1}}

\newcommand\CONSISTENT{\mathbf{Con}}
\newcommand\CONSISTENTH{\mathbf{Con}_H}

\newcommand\DL{d\mathcal{L}}

\newcommand\QUERYSATAND{\mathit{querySat}}
\newcommand\QUERYSATANDCONT{\mathit{querySat}_{C}}
\newcommand\SAT{\mathit{Sat}}
\newcommand\UNSAT{\mathit{Unsat}}

\newcommand\FORMULAS{\mathit{Formulas}}
\newcommand\OTHERWISE{\mathit{Otherwise}}

\newcommand\SPACEEX{\textsc{SpaceEx}}
\newcommand\DT{\mathit{dt}}

\pgfrealjobname{main}

\begin{document}
%
% \title{Adaptation of Generalized Property-Directed Reachability to Hybrid Systems with \\Differential Dynamic Logic Predicates}
\title{Generalized Property-Directed Reachability for Hybrid Systems}
\titlerunning{GPDR for hybrid systems}
\author{Kohei Suenaga\inst{1,2}\orcidID{0000-0002-7466-8789} \and Takuya Ishizawa\inst{1}}
%
% \authorrunning{K.~Suenaga et al.}
\institute{Kyoto University, Kyoto, Japan \and
JST PRESTO, Tokyo, Japan}
\maketitle              % typeset the header of the contribution

\begin{abstract}
  \emph{Generalized property-directed reachability} (GPDR) belongs to the family of the model-checking techniques called IC3/PDR.
  It has been successfully applied to software verification; for example, it is the core of Spacer, a state-of-the-art Horn-clause solver bundled with Z3.
  However, it has yet to be applied to hybrid systems, which involve a continuous evolution of values over time.
  As the first step towards GPDR-based model checking for hybrid systems, this paper formalizes $\HYBRIDPDR$, an adaptation of GPDR to hybrid systems, and proves its soundness.
  We also implemented a semi-automated proof-of-concept verifier, which allows a user to provide hints to guide verification steps.
  
\end{abstract}

\ifdraft
\begin{itemize}
\item Submissions are restricted to 20 pages in Springer’s LNCS format, not counting references. Additional material may be placed in an appendix, to be read at the discretion of the reviewers and to be omitted in the final version. Formatting style files and further guidelines for formatting can be found at the Springer website.
\item Deadline: October 1 (at any time zone.)

\item Revise RemoveTrace and imperativeProc.tex
\end{itemize}
\else
\fi

\section{Introduction}

% \todo{PDR $\rightarrow$ GPDR}.

A \emph{hybrid system} is a dynamical system that exhibits both continuous-time dynamics (called a \emph{flow}) and discrete-time dynamics (called a \emph{jump}).
This combination of flows and jumps is an essential feature of \emph{cyber-physical systems (CPS)}, a physical system governed by software.
In the modern world where safety-critical CPS are prevalent, their correctness is an important issue.

\emph{Model checking}~\cite{Clarke:2000:MC:332656,DBLP:journals/sttt/HenzingerHW97} is an approach to guaranteeing hybrid system safety.
It tries to prove that a given hybrid system does not violate a specification by abstracting its behavior and by exhaustively checking that the abstracted model conforms to the specification.

In the area of software model checking, an algorithm called \emph{property-directed reachability (PDR)}, also known as \emph{IC3}, is attracting interest~\cite{DBLP:conf/cav/BirgmeierBW14,DBLP:conf/tacas/CimattiGMT14,DBLP:conf/vmcai/Bradley11}.
IC3/PDR was initially proposed in the area of hardware verification; it was then transferred to software model checking by Cimatti et al.~\cite{DBLP:conf/cav/CimattiG12}.
Its effectiveness for software model checking is now widely appreciated.
For example, the SMT solver Z3~\cite{DBLP:conf/tacas/MouraB08} comes with a Horn-clause solver Spacer~\cite{DBLP:conf/cav/HoderBM11} that uses PDR internally; Horn-clause solving is one of the cutting-edge techniques to verify functional programs~\cite{DBLP:conf/birthday/BjornerGMR15,DBLP:conf/tacas/ChampionC0S18,DBLP:conf/sas/HashimotoU15} and programs with loops~\cite{DBLP:conf/birthday/BjornerGMR15}.

We propose a model checking method for hybrid automata~\cite{Alur:1993:HAA:646874.709849} based on the idea of PDR; the application of PDR to hybrid automata is less investigated compared to its application to software systems.
Concretely, we propose an adaptation of a variant of PDR called \emph{generalized property-directed reachability (GPDR)} proposed by Hoder and Bj{\o}rner~\cite{DBLP:conf/sat/HoderB12}.
%
% We exploit the advantage of the generality of GPDR.
% % that it is defined in such a way that abstracts the detail of the underlying dynamics of a verified system.  % by a forward predicate transformer.
% %
% Concretely, GPDR is parametrized over a map over predicates on states (i.e., a \emph{forward predicate transformer}); the detail of the underlying dynamics of a verified system is encapsulated into the forward predicate transformer.
% %
Unlike the original PDR, which is specialized to jump-only automata-based systems, GPDR is parametrized over a map over predicates on states (i.e., a \emph{forward predicate transformer}); the detail of the underlying dynamic semantics of a verified system is encapsulated into the forward predicate transformer.
This generality of GDPR enables the application of PDR to systems outside the scope of the original PDR by itself; for example, Hoder et al.~\cite{DBLP:conf/sat/HoderB12} show how to apply GPDR to programs with recursive function calls.

An obvious challenge in an adaptation of GPDR to hybrid automata is how to deal with flow dynamics that do not exist in software systems.
%
% More specifically, we need to define a forward predicate transformer that can express the flow dynamics of hybrid automata.
%
To this end, we extend the logic on which the forward predicate transformer is defined so that it can express flow dynamics specified by an ordinary differential equation (ODE).
Our extension, inspired by the differential dynamic logic ($\DL$) proposed by Platzer~\cite{DBLP:journals/jar/Platzer08}, is to introduce \emph{continuous reachability predicates (CRP)} of the form $\CONTIREACHPRED{\ODE}{\varphi_I}\varphi$ where $\ODE$ is an ODE and $\varphi_I$ and $\varphi$ are predicates.
This CRP is defined to hold under valuation $\sigma$ if there is a continuous transition from $\sigma$ to certain valuation $\sigma'$ that satisfies the following conditions: (1) the continuous transition is a solution of $\ODE$, (2) the valuation $\sigma'$ makes $\varphi$ true, and (3) $\varphi_I$ is true at every point on the continuous transition.
With this extended logic, we define a forward predicate transformer that faithfully encodes the behavior of a hybrid automaton.
We find that we can naturally extend GPDR to hybrid automata by our predicate transformer.

% To this end, we use \emph{differential dynamic logic ($\DL$)}~\cite{DBLP:journals/jar/Platzer08} proposed by Platzer.
% %
% $\DL$ is an extension of dynamic logic~\cite{DBLP:journals/sigact/HarelKT01}---a logic that can express and reason about the behavior of programs.
% %
% $\DL$ extends the dynamic logic so that it can express flow dynamics specified by ODE.

We formalize our adaptation of GPDR to hybrid automata, which we call $\HYBRIDPDR$.
In the formalization, we define a forward predicate transformer that precisely expresses the behavior of hybrid automata~\cite{Alur:1993:HAA:646874.709849} using $\DL$.
We prove the soundness of $\HYBRIDPDR$.
We also describe our proof-of-concept implementation of $\HYBRIDPDR$ and show how it verifies a simple hybrid automaton with human intervention.

In order to make this paper self-contained, we detail GPDR for discrete-time systems before describing our adaptation to hybrid automata.
After fixing the notations that we use in Section~\ref{sec:preliminary}, we define a discrete-time transition system and hybrid automata in Section~\ref{sec:language}. %; we also state the problem that we are to solve in this section.
Section~\ref{sec:vanillapdr} then reviews the GPDR procedure.
%
% Section~\ref{sec:hybridpdr} presents $\HYBRIDPDR$, our adaptation of GPDR to hybrid systems in a declarative style and states the soundness of the procedure.
Section~\ref{sec:hybridpdr} presents $\HYBRIDPDR$, our adaptation of GPDR to hybrid automata, and states the soundness of the procedure.
%
% For implementation purpose, an operational presentation rather than a declarative one is more convenient, which is shown in Section~\ref{sec:imperativeProc}.
%
We describe a proof-of-concept implementation in Section~\ref{sec:implementation}.
After discussing related work in Section~\ref{sec:relatedwork}, we conclude in Section~\ref{sec:conclusion}.

For readability, several definitions and proofs are presented in the appendices.

\section{Preliminary}
\label{sec:preliminary}

We write $\REAL$ for the set of reals. % and $\NONNEGREAL$ for the set of nonnegative reals.  
We fix a finite set $\vars := \set{\var_1,\dots,\var_N}$ of \emph{variables}.
We often use primed variables $\var'$ and $\var''$.  The prime notation also applies to a set of variables; for example, we write $\vars'$ for $\set{\var_1',\dots,\var_N'}$.
We use metavariable $\vec{x}$ for a finite sequence of variables.
We write $\setfml$ for the set of quantifier-free first-order formulas over $\vars \cup \vars' \cup \vars''$; its elements are ranged over by $\fml$.  
We call elements of the set $\VALUATIONS \DEFEQ (\vars \cup \vars' \cup \vars'') \ra \REAL$ a \emph{valuations}; they are represented by metavariable $\valuation$.
We use the prime notation for valuations.  For example, if $\valuation \in \vars \ra \REAL$, then we write $\valuation'$ for $\set{x_1' \mapsto \valuation(x_1), \dots, x_N' \mapsto \valuation(x_N)}$.
We write $\valuation[\var \mapsto r]$ for the valuation obtained by updating the entry for $\var$ in $\valuation$ with $r$.
We write $\valuation \models \fml$ if $\valuation$ is a model of $\fml$; $\valuation \not\models \fml$ if $\valuation \models \fml$ does not hold; $\models \fml$ if $\valuation \models \fml$ for any $\valuation$; and $\not\models \fml$ if there exists $\valuation$ such that $\valuation \not\models \fml$.
We sometimes identify a valuation $\valuation$ with a logical formula $\bigwedge_{x \in \vars}x = \valuation(x)$.

\section{State-transition systems and verification problem}
\label{sec:language}

% We are going to review the original GPDR for discrete-time systems~\cite{DBLP:conf/sat/HoderB12} in Section~\ref{sec:vanillapdr} before presenting our adaptation for hybrid systems in Section~\ref{sec:hybridpdr}.
%
% This section defines the models used in these explanations (Section~\ref{sec:dtsts} and \ref{sec:hsts}) and formally states the verification problem that we tackle (Section~\ref{sec:verificationProblem}).

We review the original GPDR for discrete-time systems~\cite{DBLP:conf/sat/HoderB12} in Section~\ref{sec:vanillapdr} before presenting our adaptation for hybrid systems in Section~\ref{sec:hybridpdr}.
This section defines the models used in these explanations (Section~\ref{sec:dtsts} and \ref{sec:hsts}) and formally states the verification problem that we tackle (Section~\ref{sec:verificationProblem}).

\subsection{Discrete-time state-transition systems (DTSTS)}
\label{sec:dtsts}

We model a discrete-time program by a state-transition system.

\begin{definition}
  \label{def:dtsts}
    A \emph{discrete-time state-transition system (DTSTS)} is a tuple $\tuple{\states,\st_0,\fml_0,\trans}$. 
    We use metavariable $\dtsts$ for DTSTS.
    $\states = \set{\st_0,\st_1,\st_2,\dots}$ is a set of \emph{locations}.  
    $\st_0$ is the initial location.
    $\fml_0$ is the formula that has to be satisfied by the initial valuation.
    $\trans \subseteq \states \times \setfml \times \setfml \times \states$ is the \emph{transition relation}.
    %
    % Each element in $\trans$ is associate with a formula $\fml$ that specifies the \emph{guard condition} for the transition to be enabled and a command $\cmd$ that specifies the \emph{action} conducted in the transition.
    %
    We write $\tuple{\st,\valuation_1} \RED{\trans} \tuple{\st',\valuation_2}$ if $\tuple{\st, \fml, \cmd, \st'} \in \trans$ where $\valuation_1 \models \fml$ and $\valuation_1 \cup \valuation_2' \models \cmd$; we call relation ${\RED{\trans}}$ the \emph{jump transition}.
    A \emph{run} of a DTSTS $\tuple{\states,\st_0,\fml_0,\trans}$ is a finite sequence $\tuple{\st^0,\valuation_0},\tuple{\st^1,\valuation_1},\dots,\tuple{\st^N,\valuation_N}$ where (1) $\st^0 = \st_0$, (2) $\valuation_0 \models \fml_0$, and (3) $\tuple{\st^i,\valuation_i} \RED{\trans} \tuple{\st^{i+1},\valuation_{i+1}}$ for any $i \in [0,N-1]$.
\end{definition}

$\tuple{\st,\varphi,\cmd,\st'} \in \trans$ intuitively means that, if the system is at the location $q$ with valuation $\valuation_1$ and $\valuation_1 \models \varphi$, then the system can make a transition to the location $\st'$ and change its valuation to $\valuation_2'$ such that $\valuation_1 \cup \valuation_2' \models \cmd$.
We call $\varphi$ the \emph{guard} of the transition.
$\cmd$ is a predicate over $\vars \cup \vars'$ that defines the \emph{command} of the transition; it defines how the value of the variables may change in this transition.
The elements of $\vars$ represent the values before the transition whereas those of $\vars'$ represent the values after the transition.
%
% For example, if $\cmd$ is $x' = x + 1$, then the value of $x$ is incremented by one by this transition.

% \begin{figure}[t]
% \begin{center}
% \begin{tikzpicture}[shorten >=1pt,node distance=2cm,on grid,auto]
%   \node[state] (q0)   {$q_0$}; 
%   \node[state] (q1) [right=of q0] {$q_1$}; 
% %   \node[state] (q_2) [right=of q_1] {$q_2$}; 
% %   \node[state](q_3) [right=of q_2] {$q_3$};
%     \path[->] 
%     (q0) edge  [bend left] node {} (q1);
%     (q1) edge  [bend right] node {} (q0);
%     (q1) edge  [loop below] node {} ();
%     (q0) edge  [loop below] node {} ();
%     % (q_2) edge  node {} (q_1) 
%     % (q_3) edge [bend left] node {} (q_1);
% \end{tikzpicture}
% \end{center}
%     \caption{Caption}
%     \label{fig:my_label}
% \end{figure}
  
\iffull
  \begin{figure}[t]
    \begin{center}
      % \beginpgfgraphicnamed{main-figure-1}%
      % \begin{tikzpicture}[shorten >=1pt,node distance=2cm,on grid,auto,
      %     initial text= \mbox{$x \geq 0 \;\land\; sum = 0$}] 
      %   \node[state, initial] (\st_0)   {$\st_0$}; 
      %   \node[state] (\st_1) [right=of \st_0] {$\st_1$}; 
      %   \path[->] 
      %   (\st_0) edge  [loop above] node
      %        {$\begin{array}{c}
      %            x > 0 \\
      %            sum := sum + x \\
      %            x := x - 1
      %        \end{array}$} ()
      %   (\st_0) edge  node {$x \leq 0$} (\st_1);
      % \end{tikzpicture}
      % \endpgfgraphicnamed
      \includegraphics{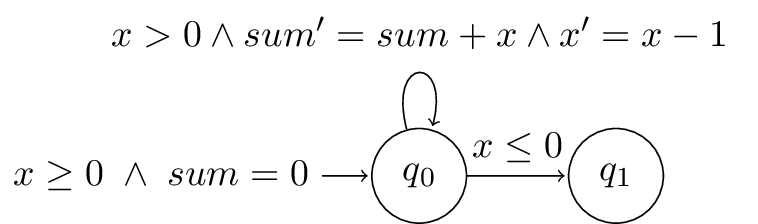}
    \end{center}
  \caption{An example of DTSTS}
  \label{figure:exampleDTSTS}
  \end{figure}
\fi

\iffull
\begin{example}
  \label{ex:dtsts}
  \begin{sloppypar}
  Figure~\ref{figure:exampleDTSTS} is an example of a DTSTS that
  models a program to compute the value of $1 + \dots + x$;
  $\states := \{\st_0,\st_1\}$ and $\fml_0 := x \ge 0 \land sum = 0$.
  In the transition from $\st_0$ to $\st_0$, the guard is $x > 0$;
  the command is $sum' = sum + x \land x' = x - 1$.
  % Intuitively, a primed variable represents the value of the original variable after a transition.
  In the transition from $\st_0$ to $\st_1$, the guard is $x \leq 0$;
  the command is $x' = x \land sum' = sum$ because this transition does not change the value of $x$ and $sum$.
  Therefore, the transition relation $\trans = \{\tuple{\st_0, x > 0, sum' = sum + x \land x' = x - 1, \st_0}, \tuple{\st_0, x \leq 0, x' = x \land sum' = sum, \st_1}\}$.
%    \todo{Example.}
% \end{example}
% % We define the set of \emph{runs} of a DTSTS as follows.
% % \begin{definition}
% %     \label{def:dtstsrun}
% % \end{definition}
% % \begin{definition}
% %     We say that a pair $\tuple{q,\valuation}$ is \emph{reachable} in $\dtsts$ if there is a run of $\dtsts$ in the middle of which $\tuple{q,\valuation}$ appears.
% %     %
% % \end{definition}
% \begin{example}
  The finite sequence $\tuple{\st_0,\{x \mapsto 3,sum \mapsto 0\}},\tuple{\st_0, \{x \mapsto 2, sum \mapsto 3\}},\\ \tuple{\st_0, \{x \mapsto 1, sum \mapsto 5\}}, \tuple{\st_0,\{x \mapsto 0, sum \mapsto 6\}},\tuple{\st_1, \{x \mapsto 0, sum \mapsto 6\}}$ is a run of the DTSTS figure~\ref{figure:exampleDTSTS}.
  % \todo{Example}
  \end{sloppypar}
\end{example}
\fi

\subsection{Hybrid automaton (HA)}
\label{sec:hsts}

We model a hybrid system by a hybrid automaton (HA)~\cite{Alur:1993:HAA:646874.709849}.
We define an HA as an extension of DTSTS as follows.

% We then extend the definition of DTSTS and its runs with continuous-time dynamics so that we can model hybrid systems.

\begin{definition}
    A \emph{hybrid automaton (HA)} is a tuple $\tuple{\states,\st_0,\fml_0,\flow,\inv,\trans}$.
    The components $\states$, $\st_0$, $\fml_0$, and $\trans$ are the same as Definition~\ref{def:dtsts}.
    We use metavariable $\hsts$ for HA.
    $\flow$ is a map from $\states$ to ODE on $\vars$ that specifies the flow dynamics at each location;
    $\inv$ is a map from $\states$ to $\setfml$ that specifies the \emph{stay condition}\footnote{We use the word "stay condition" instead of the standard terminology "invariant" following Kapur et al.~\cite{DBLP:conf/ftrtft/KapurHMP94}} at each state.
\end{definition}

A state of a hybrid automaton is a tuple $\tuple{\st,\valuation}$.
A run of $\tuple{\states,\st_0,\fml_0,\flow,\inv,\trans}$ is a sequence of states $\tuple{\st_0,\valuation_0}\tuple{\st_1,\valuation_1}\dots\tuple{\st_n,\valuation_n}$ where $\valuation_0 \models \fml_0$.
The system is allowed to make a transition from $\tuple{\st_i,\valuation_i}$ to $\tuple{\st_{i+1},\valuation_{i+1}}$ if (1) $\valuation_i$ reaches a valuation $\valuation'$ along with the flow dynamics specified by $\flow(\st_i)$, (2) $\inv(\st_i)$ holds at every point on the flow, and (3) $\tuple{\st_i,\valuation'}$ can jump to $\tuple{\st_{i+1},\valuation_{i+1}}$ under the transition relation $\trans$.
In order to define the set of runs formally, we need to define the continuous-time dynamics that happens within each location.
\begin{definition}
  \begin{sloppypar}
    Let $\ODE$ be an ordinary differential equation (ODE) on $\vars$ and let $x_1(t),\dots,x_n(t)$ be a solution of $\ODE$ where $t$ is the time.
    Let us write $\sigma^{(t)}$ for the valuation $\set{x_1 \mapsto x_1(t), \dots, x_n \mapsto x_n(t)}$.
    We write $\valuation \CONTIREACH{\ODE}{\varphi} \valuation'$ if (1) $\valuation = \valuation^{(0)}$ and (2) there exists $t' \ge t$ such that $\sigma' = \sigma^{(t')}$ and $\valuation^{(t'')} \models \varphi$ for any $t'' \in (0,t']$.
    We call relation ${\rightarrow_{\ODE,\varphi}}$ the \emph{flow transition}.
  \end{sloppypar}
\end{definition}

\begin{wrapfigure}{L}[0pt]{0.3\linewidth}
  \begin{center}
    \vspace{-1cm}
    \scalebox{0.6}{%
      % \beginpgfgraphicnamed{main-figure-2}%
      % \begin{tikzpicture}% [shorten >=1pt,node distance=4cm,on grid,auto]
      %   \node[state] (q_0)   
      %   {$\begin{array}{c}
      %       q_0 \\
      %       \dot{x} = -y \\
      %       \dot{y} = x \\
      %       y \ge 0
      %     \end{array}$}; 
      %   \node[state] (q_1) [right=of q_0]
      %   {$\begin{array}{c}
      %       q_1 \\
      %       \dot{x} = -y \\
      %       \dot{y} = x \\
      %       y \le 0
      %     \end{array}$}; 
      %   \path[->] 
      %   (q_0) edge [bend left] node [above] {$y \leq 0$} (q_1)
      %   (q_1) edge [bend left] node [below] {$y \geq 0$}(q_0);
      % \end{tikzpicture}
      % \endpgfgraphicnamed
      \includegraphics{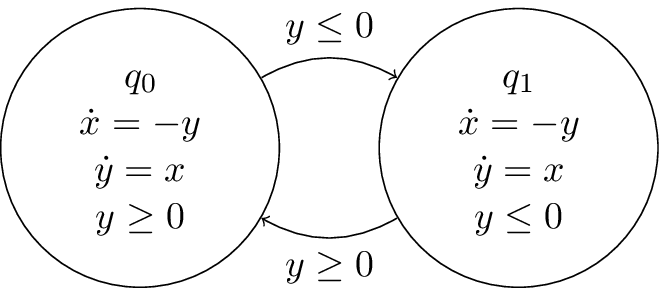}
  }
  \vspace{-1cm}
  \end{center}
  \caption{An example of a hybrid automaton.}
  \label{fig:exampleHA}
  \vspace{-0.9cm}
\end{wrapfigure}

Intuitively, the relation $\valuation \CONTIREACH{\ODE}{\varphi} \valuation'$ means that there is a trajectory from the state represented by $\valuation$ to that represented by $\valuation'$ such that (1) the trajectory is a solution of $\ODE$ and (2) $\varphi$ holds at any point on the trajectory.
For example, let $\ODE$ be $\dot{x} = v, \dot{v} = 1$, where $x$ and $v$ are time-dependent variables; $\dot{x}$ and $\dot{v}$ are their time derivative.
The solution of $\ODE$ is $v = t + v_0$ and $x = \frac{t^2}{2} + v_0 t + x_0$ where $t$ is the elapsed time, $x_0$ is the initial value of $x$, and $v_0$ is the initial value of $v$.
Therefore, $\set{x \mapsto 0, v \mapsto 0} \CONTIREACH{\ODE}{\TRUE} \set{x \mapsto \frac{1}{2}, v \mapsto 1}$ holds because $(x,v) = (\frac{1}{2}, 1)$ is the state at $t = 1$ on the above solution with $x_0 = 0$ and $v_0 = 0$.
$\set{x \mapsto 0, v \mapsto 0} \CONTIREACH{\ODE}{x \ge 0} \set{x \mapsto \frac{1}{2}, v \mapsto 1}$ also holds because the condition $x \ge 0$ continues to hold along with the trajectory from $(x,v) = (0,0)$ to $(\frac{1}{2},1)$.
However, $\set{x \mapsto 0, v \mapsto 0} \CONTIREACH{\ODE}{x \ge \frac{1}{4}} \set{x \mapsto \frac{1}{2}, v \mapsto 1}$ does \emph{not} hold because the condition $x \ge \frac{1}{4}$ does not hold for the initial $\frac{1}{\sqrt{2}}$ seconds in this trajectory.

Using this relation, we can define a run of an HA as follows.
\begin{definition}
  \label{def:harun}
  A finite sequence $\tuple{\st^0,\valuation_0},\tuple{\st^1,\valuation_1},\dots,\tuple{\st^N,\valuation_N}$ is called a \emph{run} of an HA $\tuple{\states,\st_0,\fml_0,\flow,\inv,\trans}$ if (1) $\st^0 = \st_0$, (2) $\valuation_0 \models \fml_0$, (3) for any $i$, if $0 \le i \le N-2$, there exists $\tuple{\st^i, \fml_i, \cmd^i, \st^{i+1}} \in \trans$ and $\valuation^I$ such that $\valuation_i \CONTIREACH{\flow(q^i)}{\inv(q^i)} \valuation^I$ and $\valuation^I \models \fml_i$ and $\tuple{\st^i,\valuation^I} \RED{\trans} \tuple{\st^{i+1},\valuation_{i+1}}$, and (4) $\valuation_{N-1} \CONTIREACH{\flow(q^{N-1})}{\inv(\st^{N-1})} \valuation_{N}$.
\end{definition}

\begin{remark}
\label{remark:lasttransofha}
This definition is more complicated than that of runs of DTSTS because we need to treat the last transition from $\tuple{\st^{N-1},\valuation_{N-1}}$ to $\tuple{\st^{N},\valuation_{N}}$ differently than the other transitions.
Each transition from $\tuple{\st^i,\valuation_i}$ to $\tuple{\st^{i+1},\valuation_{i+1}}$, if $0 \le i \le N-2$, is a flow transition followed by a jump transition; however, the last transition consists only of a flow transition.
\end{remark}

\begin{example}
  Figure~\ref{fig:exampleHA} shows a hybrid automaton with $\states := \set{\st_0,\st_1}$ schematically.
  Each circle represents a location $\st$; we write $\flow(\st)$ for the ODE associated with each circle.
  Each edge between circles represents a transition; we present the guard of the transition on each edge.
  We omit the $\varphi_c$ part; it is assumed to be the do-nothing command represented by $\land_{x \in \vars}x' = x$.

  Both locations are equipped with the same flow that is the anticlockwise circle around the point $(x,y)=(0,0)$ on the $xy$ plane.
  The system can stay at $\st_0$ as long as $y \ge 0$ and at $\st_1$ as long as $y \le 0$.
  $y=0$ holds whenever a transition is invoked.
  Indeed, for example, $\inv(\st_0) = y \ge 0$ and the guard from $\st_0$ to $\st_1$ is $y \le 0$; therefore, when the transition is invoked, $\inv(\st_0) \land y \le 0$ holds, which is equivalent to $y=0$.

  Starting from the valuation $\valuation_0 := \set{x \mapsto 1, y \mapsto 0}$ at location $\st_0$, the system reaches $\valuation_1 := \set{x \mapsto -1, y \mapsto 0}$ by the flow $\flow(\st_0)$ along which $\inv(\st_0) \equiv y \ge 0$ continues to hold; then the transition from $\st_0$ to $\st_1$ is invoked.
  After that, the system reaches $\valuation_2 := \set{x \mapsto 0, y \mapsto -1}$ by $\flow(\st_1)$.
  Therefore, $\tuple{\st_0,\valuation_0}\tuple{\st_1,\valuation_1}\tuple{\st_1,\valuation_2}$ is a run of this HA.
\end{example}

% \begin{remark}
%     An HA is essentially equivalent to a hybrid automaton~\cite{}.
%     \todo{Explanation}
% \end{remark}

\subsection{Safety verification problem}
\label{sec:verificationProblem}

% Now that we defined DTSTS and HA, the verification problem we are to address can be formally stated as follows.
\begin{definition}
  We say that $\valuation$ is \emph{reachable} in DTSTS $\dtsts$ (resp., HA $\hsts$) if there is a run of $\dtsts$ (resp., $\hsts$) that reaches $\tuple{\st,\valuation}$ for some $\st$.
  A \emph{safety verification problem (SVP)} for a DTSTS $\tuple{\dtsts,\varphi}$ (resp., HA $\tuple{\hsts,\varphi}$) is the problem to decide whether $\valuation' \models \varphi$ holds for all the reachable valuation $\valuation'$ of the given $\dtsts$ (resp., $\hsts$).
\end{definition}

If an SVP is affirmatively solved, then the system is said to be \emph{safe}; otherwise, the system is said to be \emph{unsafe}.
One of the major strategies for proving the safety of a system is discovering its \emph{inductive invariant}.
\begin{definition}
  \label{def:inductiveInvariants}
  \begin{itemize}
  \item
    \begin{sloppypar}
    Let $\tuple{\dtsts,\varphi_P}$ be an SVP for DTSTS where $\dtsts = \tuple{\states,\st_0,\fml_0,\trans}$.
    Then, a function $\rel \COL \states \ra \setfml$ is called an inductive invariant if (1) $\models \fml_0 \implies \rel(\st_0)$; (2) if $\valuation \models \rel(\st)$ and $\tuple{\st,\valuation} \RED{\trans} \tuple{\st',\valuation'}$, then $\valuation' \models \rel(\st')$; and (3) $\models \rel(\st) \implies \varphi_P$ for any $\st$.
    \end{sloppypar}
  \item
    Let $\tuple{\hsts,\varphi_P}$ be an SVP for HA where $\hsts = \tuple{\states,\st_0,\fml_0,\flow,\inv,\trans}$.
    Then, a function $\rel \COL \states \ra \setfml$ is called an inductive invariant if (1) $\models \fml_0 \implies \rel(\st_0)$; (2) if $\valuation \models \rel(\st)$ and $\tuple{\st,\valuation} \CONTIREACH{\flow(\st)}{\inv(\st)} \tuple{\st'',\valuation''}$ and $\tuple{\st'',\valuation''} \RED{\trans} \tuple{\st',\valuation'}$, then $\valuation' \models \rel(\st')$; and (3) $\models \rel(\st) \implies \varphi_P$ for any $\st$.
  \end{itemize}
\end{definition}

% \begin{lemma}
%   \label{lem:invariantImpliesSafety}
%   If there exists an inductive invariant $\rel$ for an SVP $\tuple{\dtsts,\varphi_P}$ for DTSTS where $\dtsts = \tuple{\states,\st_0,\fml_0,\trans}$, then the SVP is affirmatively solved.
% \end{lemma}

% An SVP is solved negatively if a \emph{counterexample} is found.

Unsafety can be proved by discovering a \emph{counterexample}.
\begin{definition}
  Define $\dtsts$, $\varphi_P$, and $\hsts$ as in Definition~\ref{def:inductiveInvariants}.
  A run $\tuple{\valuation_0,\st_0}\dots\tuple{\valuation_N,\st_N}$ of $\dtsts$ (resp. $\hsts$) is called a counterexample to the SVP $\tuple{\dtsts,\varphi_P}$ (resp. $\tuple{\hsts,\varphi_P}$) if $\valuation_N \models \neg\varphi_P$.
\end{definition}

GPDR is a procedure that tries to find an inductive invariant or a counterexample to a given SVP.
% \todo{Mention that PDR tries to find an inductive invariant to show safety.}
SVP is in general undecidable.
Therefore, the original GPDR approach~\cite{DBLP:conf/sat/HoderB12} and our extension with hybrid systems presented in Section~\ref{sec:hybridpdr} do not terminate for every input.

\section{GPDR for DTSTS}
\label{sec:vanillapdr}
Before presenting our extension of GPDR with hybrid systems, we present the original GPDR procedure by Hoder and Bj{\o}rner~\cite{DBLP:conf/sat/HoderB12} in this section.
(The GPDR presented here, however, is slightly modified from the original one; see Remark~\ref{remark:localation-specific-gpdr}.)
%
% There are several variants of the presentation of PDR.
%
% We follow the presentation of \emph{Generalized PDR (GPDR)} 
%
% Their presentation, compared to the original presentation of PDR by Bradley~\cite{DBLP:conf/sat/Bradley12}, is not specialized to a specific description of underlying systems, which makes our adaptation to hybrid systems easier.
% in less operational manner, which makes we believe the overall idea easier to understand.
%
% We note that the presentation of GPDR in this paper is slightly modified from the original one  so that it fits our definition of DTSTS.

Given a safety verification problem $\tuple{\dtsts,\fml_P}$ where $\dtsts = \tuple{\states,\st_0,\fml_0,\trans}$, GPDR tries to find (1) an inductive invariant to prove the safety of $\dtsts$, or (2) a counterexample to refute the safety.
To this end, GPDR (nondeterministically) manipulates a data structure called \emph{configurations}.
A configuration is either $\RESVALID$, $\RESMODEL \cetrace$, or an expression of the form $\PDRState{\cetrace}{R_0,\dots,R_N; N}$.
We explain each component of the expression $\PDRState{\cetrace}{R_0,\dots,R_N; N}$ in the following.
($\RESVALID$ and $\RESMODEL \cetrace$ are explained later.)
\begin{itemize}
    \item $R_0,\dots,R_N$ is a finite sequence of maps from $\states$ to $\setfml$ (i.e., elements of $\setfml$).
    Each $R_i$ is called a \emph{frame}.
    %
    % Let us write $\varphi_{i,j}$ for $R_i(\st_j)$; then
    The frames are updated during an execution of GPDR so that $R_i(\st_j)$ is an overapproximation of the states that are reachable within $i$ steps from the initial state in $\dtsts$ and whose location is $\st_j$.
    \item $N$ is the index of the last frame.
    \item
      $\cetrace$ is a finite sequence of the form $\tuple{\valuation_i,\st_i,i},\tuple{\valuation_i,\st_i,i+1},\dots,\tuple{\valuation_N,\st_N,N}$.
      This sequence is a candidate partial counterexample that starts from the one that is $i$-step reachable from the initial state and that ends up with a state $\tuple{\valuation_N,\st_N}$ such that $\valuation_N \models \neg\varphi_P$.
    %   $\set{\tuple{\valuation_i,\st_i,k_i},\dots,\tuple{\valuation_N,\st_N,k_N}}$
    %   Each $\tuple{\valuation_i,\st_i,k_i}$ represents a \emph{candidate} counterexample found in frame $R_{k_i}$ at location $\st_i$.
    %   %
    %   If $\tuple{\valuation_i,\st_i,k_i} \in \cetrace$, then there is a run of $\dtsts$ starting from a state $\tuple{\st_i,\valuation_i}$ and ending up in a state that satisfies $\neg\varphi_P$.
    % %
      Therefore, in order to prove the safety of $\dtsts$, a GPDR procedure needs to prove that $\tuple{\st_i,\valuation_i}$ is unreachable within $i$ steps from an initial state.
\end{itemize}

In order to formalize the above intuition, GPDR uses a \emph{forward predicate transformer} determined by $\dtsts$.
In the following, we fix an SVP $\tuple{\dtsts,\fml_P}$.
\begin{definition}
\label{def:predtransOrig}
$\predtrans(R)(\st')$, where $\predtrans$ is called the \emph{forward predicate transformer} determined by $\dtsts$, is the following formula:
  \[
    \begin{array}{l}
      (\st' = \st_0 \land \varphi_0) \lor \displaystyle\bigvee_{ (\st,\varphi,\cmd,\st') \in \trans} \exists \vec{x''}.
      \left(
      \begin{array}{ll}
        & [\vec{x''}/\vec{x}]R(\st)\\
        \land & [\vec{x''}/\vec{x}]\varphi \land [\vec{x}/\vec{x'},\vec{x''}/\vec{x}]\cmd
      \end{array}
      \right),
    \end{array}
  \]
  where $\vec{x''}$ is the sequence $\var_1'',\dots,\var_N''$.
  % wher $T_c$ is the formula defined as follows:
  % \[
  %   \begin{array}{rcl}
  %     T_{\SKIP}  &:=& \displaystyle\bigwedge_{x \in \vars} (x = x')\\
  %     T_{x := e} &:=& x' = e \land \displaystyle\bigwedge_{y \in \vars\setdiff\set{x}}(y = y')\\
  %     T_{c_1 \SEQ c_2} &:=& \exists \vec{x''}. [\vec{x''}/\vec{x'}]T_{c_1} \land [\vec{x'}/\vec{x''}]T_{c_2}\\
  %   \end{array}
  % \]

\end{definition}

Notice that $\predtrans(\lambda \st. \FALSE)$ is equivalent to $\varphi_0$.
Intuitively, $\valuation' \models \predtrans(R)(\st')$ holds if $\tuple{\st',\valuation'}$ is an initial state (i.e., $\st'=\st_0$ and $\valuation' \models \varphi_0$) or $\tuple{\st',\valuation'}$ is reachable in 1-step transition from a state that satisfies $R$.
The latter case is encoded by the second disjunct of the above definition: The valuation $\valuation'$ satisfies the second disjunct if there are $\st,\varphi$, and $\varphi_c$ such that $(\st,\varphi,\varphi_c,\st') \in \trans$ (i.e., $\st'$ is 1-step after $\st$ in $\delta$) and there is a valuation $\valuation$ such that $\valuation \models R(\st) \land \varphi$ (i.e., $\valuation$ satisfies the precondition $R(\st)$ and the guard $\varphi$) and $\valuation'$ is a result of executing command $c$ under $\valuation$.

The following lemma guarantees that $\predtrans$ soundly approximates the transition of an DTSTS.
\begin{lemma}
  \label{lem:predtransProp}
  If $\valuation_1 \models R(\st_1)$ and $\tuple{\st_1,\valuation_1} \RED{\trans} \tuple{\st_2,\valuation_2}$, then $\valuation_2 \models \predtrans(R)(\st_2)$.
\end{lemma}
\begin{proof}
  \begin{sloppypar}
  Assume $\valuation_1 \models R(\st_1)$ and $\tuple{\st_1,\valuation_1} \RED{\trans} \tuple{\st_2,\valuation_2}$.
  Then, by definition, $(q_1,\varphi,\varphi',q_2) \in \delta$ and $\sigma_1 \models \varphi$ and $\sigma_1 \cup \sigma_2' \models \varphi_c$ for some $\varphi$ and $\varphi_c$.
  $\sigma_1'' \cup \sigma_2 \models [\vec{x''}/\vec{x}]R(\st_1)$ follows from $\sigma_1 \models R(\st_1)$.
  $\sigma_1'' \cup \sigma_2 \models [\vec{x''}/\vec{x}]\varphi$ follows from $\sigma_1 \models \varphi$.
  $\sigma_1'' \cup \sigma_2 \models [\vec{x}/\vec{x'}, \vec{x''}/\vec{x}]\varphi_c$ follows from $\sigma_1 \cup \sigma_2' \models \varphi_c$.
  Therefore, $\sigma_1'' \cup \sigma_2 \models [\vec{x''}/\vec{x}]R(\st_1) \land [\vec{x''}/\vec{x}]\varphi \land [\vec{x}/\vec{x'}, \vec{x''}/\vec{x}]\varphi_c$.
  Hence, we have $\sigma_2 \models \exists \vec{x''}. [\vec{x''}/\vec{x}]R(q) \land [\vec{x''}/\vec{x}]\varphi \land [\vec{x}/\vec{x'},\vec{x''}/\vec{x}]\varphi_c$ as required.
  \end{sloppypar}
\end{proof}

By using the forward predicate transformer $\predtrans$, we can formalize the intuition about configuration $\PDRState{\cetrace}{R_0,\dots,R_N; N}$ explained so far as follows.
\begin{definition}
  \label{def:originalCon}
Let $\dtsts$ be $\tuple{\states,\st_0,\fml_0,\trans}$, $\predtrans$ be the forward predicate transformer determined by $\dtsts$, and $\varphi_P$ be the safety condition to be verified.
A configuration $C$ is said to be \emph{consistent} if it is (1) of the form $\RESVALID$, (2) of the form $\RESMODEL \tuple{\valuation,\st_0,0}\cetrace$, or (3) of the form $\PDRState{\cetrace}{R_0,\dots,R_N; N}$ that satisfies all of the following conditions:
\begin{itemize}
\item (Con-A) $R_0(\st_0) = \fml_0$ and $R_0(\st_i) = \FALSE$ if $\st_i \ne \st_0$;
\item (Con-B) $\models R_i(\st) \implies R_{i+1}(\st)$ for any $\st$;
\item (Con-C) $\models R_i(\st) \implies \varphi_P$ for any $\st$ and $i < N$;
\item (Con-D) $\models \predtrans(R_i)(\st) \implies R_{i+1}(\st)$ for any $i < N$ and $\st$;
\item (Con-E) if $\tuple{\sigma,\st,N} \in \cetrace$, then $\sigma \models R_N(\st) \land \neg\varphi_P$\footnote{We hereafter write $\tuple{\sigma,\st,i} \in \cetrace$ to express that the element  $\tuple{\sigma,\st,i}$ exists in the sequence $\cetrace$ although $\cetrace$ is a sequence, not a set.}; and
\item (Con-F) if $\tuple{\sigma_1,\st_1,i}, \tuple{\sigma_2,\st_2,i+1} \in \cetrace$ and $i < N$, then $\tuple{\st_1,\varphi,\cmd,\st_2} \in \trans$ and $\sigma_1,\sigma_2' \models R_i(\st_1) \land \varphi \land \varphi_c$.
\end{itemize}
If $C$ is consistent, we write $\CONSISTENT(C)$.
\end{definition}

\begin{figure}[t]
  \[
    \scriptsize
    \begin{array}{lrcl}
      \INITIALIZE & & \Lra & \PDRState{\emptyset}{\tuple{\rel_0 := \predtrans(\lambda \st. \FALSE) ; N := 0}}\\
      &&&\IF \forall \st \in \states. \models \rel_0(\st) \implies \varphi_P\\
      \VALID & \PDRState{\cetrace}{\abstraction} & \Lra & \RESVALID\\
                  &&&\IF \exists i < N. \forall \st \in \states. \models \rel_{i}(\st) \implies \rel_{i-1}(\st)\\
      \UNFOLD & \PDRState{\cetrace}{\abstraction} & \Lra & \PDRState{\emptyset}{A[\rel_{N+1} := \lambda \st. \TRUE; N := N + 1]}\\
                  &&&\IF \forall \st \in \states. \models \rel_N(\st) \implies \varphi_P \\
      \INDUCTION & \PDRState{\cetrace}{\abstraction} & \Lra & \PDRState{\emptyset}{A[\rel_j := \lambda \st. \rel_j(\st) \land \rel(\st)]_{j=1}^{i+1}}\\
                  &&&\begin{array}[t]{l}\IF % (\fml \lor \psi) \in R_i, \fml \not\in R_{i+1} \\
                       \forall \st \in \states. \models \predtrans(\lambda \st. \rel_i(\st) \land \rel(\st))(\st) \implies \rel(\st)\\
                     \end{array}\\
      \CANDIDATE & \PDRState{\emptyset}{\abstraction} & \Lra & \PDRState{\tuple{\valuation,\st,N}}{\abstraction}\\
                  &&& \IF \valuation \models \rel_N(\st) \land \neg \varphi_P \\
      \DECIDE & \PDRState{\tuple{\valuation_2,\st_2,i+1}\cetrace}{\abstraction} & \Lra & \PDRState{\tuple{\valuation_1,\st_1,i}\tuple{\valuation_2,\st_2,i+1}\cetrace}{A}\\
                  &&&
                      \begin{array}[t]{l}
                        \IF
                        \tuple{\st_1,\varphi,\cmd,\st_2} \in \trans
                        \AND
                        \valuation_1,\valuation_2' \models \rel_i(\st_1) \land \varphi \land \cmd\\
                        % \mbox{where $\valuation_2$ obtained by priming every variables in $\valuation'$}\\
                        % \transrel_{\st,\st'}[\rel_i(\st)] \\
                        % \AND i = 0 \implies \st = \st_0\\
                      \end{array}\\
      \MODEL & \PDRState{\tuple{\valuation,\st_0,0}\cetrace}{A} & \Lra & \RESMODEL \tuple{\valuation,\st_0,0}\cetrace\\
      % \hline\\
      % \MODELDISC & \PDRState{\tuple{\valuation,\st,0}\cetrace}{A} & \Lra & \PDRState{\cetrace}{A}\\
      % &&&\IF \st \ne \st_0\\
      \CONFLICT & \PDRState{\tuple{\valuation',\st',i+1}\cetrace}{A} & \Lra & \PDRState{\emptyset}{A[\rel_j \la \lambda \st. \rel_j(\st) \land \rel(\st)]_{j=1}^{i+1}}\\ 
                  &&& \IF \models \rel(\st') \implies \neg\valuation' \AND \forall \st \in \states. \models \predtrans(R_i)(\st) \implies \rel(\st)\\
    \end{array}
  \]
\caption{The rules for the original PDR.  Recall that $\neg\valuation'$ in the rule $\CONFLICT$ denotes the formula $\displaystyle\neg\left(\bigwedge_{x \in \vars}x = \valuation'(x)\right)$.}
\label{fig:defVanillaPDR}
\end{figure}

% \begin{definition}
% We write $\CONSISTENT(\PDRState{\cetrace}{\rel_0,...,\rel_n})$ if 
% for $0 \leq i < n$, $\rel_i \ra \rel_{i+1}$, $\rel_i \ra \safe$ and 
% $\predtrans(\rel_i) \ra \rel_{i+1}$ hold.

% % \todo{"That" invariant of GPDR paper.}
% \end{definition}

The GPDR procedure rewrites a configuration following the (nondeterministic) rewriting rules in Figure~\ref{fig:defVanillaPDR}.
We add a brief explanation below; for more detailed exposition, see~\cite{DBLP:conf/sat/HoderB12}.
Although the order of the applications of the rules in Figure~\ref{fig:defVanillaPDR} is arbitrary, we fix one scenario of the rule applications in the following for explanation.
% \todo{Explain more after we fix the rewriting rules.}
\begin{enumerate}
\item
  \begin{sloppypar}
    The procedure initializes $\cetrace$ to $\emptyset$, $R_0$ to $\predtrans(\lambda\st.\FALSE)$, and $N$ to $0$ ($\INITIALIZE$).
  \end{sloppypar}
\item\label{code:loopentry} If there are a valuation $\valuation$ and a location $\st$ such that $\valuation \models \rel_N(\st) \land \neg \varphi_P$ ($\CANDIDATE$), then the procedure adds $\tuple{\valuation,\st,N}$ to $\cetrace$.
  The condition $\valuation \models \rel_N(\st) \land \neg \varphi_P$ guarantees that the state $\tuple{\st,\valuation}$ violates the safety condition $\varphi_P$; therefore, the candidate $\tuple{\valuation,\st,N}$ needs to be refuted.
  If not, then the frame sequence is extended by setting $N$ to $N+1$ and $\rel_{N+1}$ to $\lambda \st. \TRUE$ ($\UNFOLD$); this is allowed since $\forall \st \in \states. \models \rel_N(\st) \implies \varphi_P$ in this case.
\item
  The discovered $\tuple{\st,\valuation}$ is backpropagated by successive applications of $\DECIDE$: In each application of $\DECIDE$, for $\tuple{\st_2,\valuation_2,i+1}$ in $\cetrace$, the procedure tries to find $\valuation$ and $\st$ such that $\tuple{\st_1,\varphi,\cmd,\st_2} \in \trans$ and $\valuation_1,\valuation_2' \models \rel_i(\st_1) \land \varphi \land \cmd$ where $\valuation_2'$ is the valuation obtained by replacing the domain of $\valuation_2$ with their primed counterpart.
  These conditions in combination guarantee $\tuple{\st_1,\valuation_1} \RED{\trans} \tuple{\st_2,\valuation_2}$ and $\valuation_1 \models R_i(\st_1)$.
  \begin{enumerate}
  \item
    If this backpropagation reaches $R_0$ (the rule $\MODEL$), then it reports the trace of the backpropagation returning $\RESMODEL\ \cetrace$.
  \item
    If it does not reach $R_0$, in which case there exists $i$ such that $\valuation' \land \predtrans(R_i)(\st')$ is not satisfiable, then we pick a frame $R$ such that $\models R(\st') \implies \neg\valuation'$ and $\models \predtrans(R_i)(\st) \implies R(\st)$ for any $\st$ (the rule $\CONFLICT$).
    Intuitively, $R$ is a frame that separates (1) the union of the initial states denoted by $\varphi_0$ and the states that are one-step reachable from a state denoted by $R_i(\st')$ and (2) the state denoted by $\tuple{\st',\valuation'}$.
    In a GPDR term, $R$ is a \emph{generalization} of $\neg\valuation'$.
    This formula is used to strengthen $\rel_j$ for $j \in \set{1,\dots,i+1}$.
  \end{enumerate}
\item
  The frame $R$ obtained in the application of the rule $\CONFLICT$ is propagated forward by applying the rule $\INDUCTION$.
  The condition $\forall \st \in \states. \models \predtrans(\lambda\st.R_i(\st) \land \rel(\st))(\st) \implies \rel(\st)$ forces that $\rel$ holds in the one-step transition from a states that satisfies $R_i$.
  If this condition holds, then $\rel$ holds for $i+1$ steps (Theorem~\ref{lem:invariant}); therefore, we conjoin $\rel$ to $R_1(\st), \dots, R_{i+1}(\st)$.
  In order to maintain the consistency conditions (Con-E) and (Con-F), this rule clears $\cetrace$ to the empty set to keep its consistency to the updated frames.\footnote{We could filter $\cetrace$ so that it is consistent for the updated frame.  We instead discard $\cetrace$ here for simplicity.}
\item
  If $\forall \st \in \states. \models \rel_i(\st) \implies \rel_{i-1}(\st)$ for some $i < N$, then the verification succeeds and $\rel_i$ is an inductive invariant ($\VALID$).
  If such $i$ does not exist, then we go back to Step~\ref{code:loopentry}.
\end{enumerate}

\begin{remark}
  \label{remark:localation-specific-gpdr}
  One of the differences of the above GPDR from the original one~\cite{DBLP:conf/sat/HoderB12} is that ours deals with the locations of a given DTSTS explicitly.
  In the original GPDR, information about locations are assumed to be encoded using a variable that represents the program counter.
  Although such extension was proposed for IC3 by Lange et al.~\cite{DBLP:conf/fmcad/0001NN15}, we are not aware of a variant of GPDR that treats locations explicitely.
\end{remark}

\paragraph{Soundness. }
We fix one DTSTS $\tuple{\states,\st_0,\fml_0,\trans}$ in this section.
The correctness of the GPDR procedure relies on the following lemmas.
\begin{lemma}
  \label{lem:invariant}
$\CONSISTENT$ is invariant to any rule application of Figure~\ref{fig:defVanillaPDR}.
\end{lemma}

% \begin{lemma}
% If $\CONSISTENT(\PDRState{\cetrace}{\rel_0,...,\rel_N; N})$ holds,
% then,
% for all $i \in [0,N]$ and for any $\st_0,\dots,\st_i$ and $\valuation_0,\dots,\valuation_i$,
% if $\valuation_0 \models \varphi_0$ and
% $\tuple{\st_0,\valuation_0} \RED{\trans} \dots \RED{\trans} \tuple{\st_i,\valuation_i}$,
% then $\valuation_i \models \rel_i(\st_i)$.
% \end{lemma}

\begin{theorem}
  \label{th:vanillaSoundness}
  If the GPDR procedure is started from the rule $\INITIALIZE$ and leads to $\RESVALID$, then the system is safe.
  If the GPDR procedure is started from the rule $\INITIALIZE$ and leads to $\RESMODEL \tuple{\valuation_0,\st_0,0}\dots\tuple{\valuation_N,\st_N,N}$, then the system is unsafe.
\end{theorem}

\section{$\HYBRIDPDR$}
\label{sec:hybridpdr}

We now present our procedure $\HYBRIDPDR$ that is an adaptation of the original GPDR to hybrid systems.
An adaptation of GPDR to hybrid systems requires the following two challenges to be addressed.
\begin{enumerate}
    \item\label{challenge1}
    The original definition of $\predtrans$ (Definition~\ref{def:predtransOrig}) captures only a discrete-time transition.
    In our extension of GPDR, we need a forward predicate transformer that can mention a flow transition.
    \item\label{challenge2}
    A run of an HA (Definition~\ref{def:harun}) differs from that of DTSTS in that its last transition consists only of flow dynamics; see Remark~\ref{remark:lasttransofha}.
\end{enumerate}
In order to address the first challenge, we extend the logic on which $\predtrans$ is defined to be able to mention flow dynamics and define $\predtrans$ on the extended logic (Section~\ref{sec:hybridpredtrans}).
To address the second challenge, we extend the configuration used by GPDR so that it carries an overapproximation of the states that are reachable from the last frame by a flow transition; the GPDR procedure is also extended to maintain this information correctly (Section~\ref{sec:extendedprocedure}).

\subsection{Extension of forward predicate transformer}
\label{sec:hybridpredtrans}

In order to extend $\predtrans$ to accommodate flow dynamics, we extend the logic on which $\predtrans$ is defined with \emph{continuous reachability predicates (CRP)} inspired by the differential dynamic logic ($\DL$) proposed by Platzer~\cite{DBLP:journals/ki/Platzer10}.
\begin{definition}
  \label{def:contReachPred}
  Let $\ODE$ be an ODE over $Y := \set{y_1,\dots,y_k} \subseteq \vars$.
  Let us write $\valuation$ for $\set{y_1 \mapsto e_1, \dots, y_k \mapsto e_k}$ and $\valuation'$ for $\set{y_1 \mapsto e_1', \dots, y_k \mapsto e_k'}$.
  We define a predicate $\CONTIREACHPRED{\ODE}{\varphi}\varphi'$ by:
  $\valuation \models \CONTIREACHPRED{\ODE}{\varphi}\varphi' \mathrm{\quad iff. \quad} \exists \valuation'. \valuation \CONTIREACH{\ODE}{\varphi} \valuation' \land \valuation' \models \varphi'$.
  We call a predicate of the form $\CONTIREACHPRED{\ODE}{\varphi_I}\varphi$ a \emph{continuous reachability predicate (CRP)}.
\end{definition}

% 
% We present our logic as an extension of the original logic of $\CONTIREACHPRED{\ODE}{\varphi}\varphi'$ for the simplicity of the presentation.

Using the above predicate, we extend $\predtrans$ as follows.

\begin{definition}
\label{def:hybridpredtrans}
For an HA $\tuple{\states,\st_0,\fml_0,\flow,\inv,\trans}$, the forward predicate transformer $\predtranshybrid(R)(\st')$ is the following formula:
\[
    \begin{array}{ll}
      (\st'=\st_0 \land \varphi_0) \lor\\
      \displaystyle\bigvee_{(\st,\varphi,\cmd,\st') \in \trans} \exists \vec{x''}.\left(
        \begin{array}{ll}
          & [\vec{x''}/\vec{x}]R(\st)\\
                  % \land & [\vec{x_0}/\vec{x}]\CONTIREACHPRED{\flow(\st)}{\inv(\st)}(\vec{x} = \vec{x_1})\\
          \land & \CONTIREACHPRED{[\vec{x''}/\vec{x}]\flow(\st)}{[\vec{x''}/\vec{x}]\inv(\st)}([\vec{x''}/\vec{x}]\varphi \land [\vec{x}/\vec{x'},\vec{x''}/\vec{x}]\cmd)
        \end{array}
        \right).
    \end{array}
\]
In the above definition, $[\vec{x''}/\vec{x}]\flow(\st)$ is the ODE obtained by renaming the variables $\vec{x}$ that occur in ODE $\flow(\st)$ with $\vec{x''}$.

We also define predicate $\predtranscont(R)(\st')$ as follows:
\[
\begin{array}{l}
  % \fml_0 \lor
  \exists \vec{x''}. ([\vec{x''}/\vec{x}]R(\st') \land \CONTIREACHPRED{[\vec{x''}/\vec{x}]\flow(\st')}{[\vec{x''}/\vec{x}]\inv(\st')} \vec{x} = \vec{x''}).
\end{array}
\]
\end{definition}
Intuitively, $\valuation' \models \predtranshybrid(\fml)(\st')$ holds if either (1) $\tuple{\st',\valuation'}$ is an initial state or (2) it is reachable from $R$ by a flow transition followed a jump transition.
Similarly, $\valuation' \models \predtranscont(\rel)(\st')$ holds if $\valuation'$ is reachable in a flow transition (not followed by a jump transition) from a state denoted by $\rel(\st')$.
%
% The latter property is guaranteed by the existence of $$ a valuation $\valuation$ such that $\valuation \models \fml(\st)$
This definition of $\predtranshybrid$ is an extension of Definition~\ref{def:predtransOrig} in that it encodes the "flow-transition" part of the above intuition by the CRP.
In the case of $\predtranscont$, the postcondition part of the CRP is $\vec{x}=\vec{x''}$ because we do not need a jump transition in this case.

% A key property of the extended $\predtranshybrid$ is the following.
\begin{lemma}
  \label{lem:predtranshybridProp}
  If $\valuation_1 \models R(\st_1)$ and $\valuation_1 \CONTIREACH{\flow(q_1)}{\inv(q_1)} \valuation^I$ and $\tuple{\st_1,\valuation^I} \RED{\trans} \tuple{\st_2,\valuation_2}$, then $\valuation_2 \models \predtranshybrid(R)(\st_2)$.
\end{lemma}
\begin{proof}
  \begin{sloppypar}
  Assume (1) $\valuation_1 \models R(\st_1)$, (2) $\valuation_1 \CONTIREACH{\flow(q_1)}{\inv(q_1)} \valuation^I$, and (3) $\tuple{\st_1,\valuation^I} \RED{\trans} \tuple{\st_2,\valuation_2}$.
  Then, by definition, (4) $(q_1,\varphi,\varphi_c,q_2) \in \delta$ and (5) $\sigma^I \models \varphi$ and (6) $\sigma^I \cup \sigma_2' \models \varphi_c$ for some $\varphi$ and $\varphi_c$.
  We show $\exists \vec{x''}.\left(
[\vec{x''}/\vec{x}]R(\st) \land \CONTIREACHPRED{[\vec{x''}/\vec{x}]\flow(\st)}{[\vec{x''}/\vec{x}]\inv(\st)}([\vec{x''}/\vec{x}]\varphi \land [\vec{x}/\vec{x'},\vec{x''}/\vec{x}]\cmd) \right)$.
  (5) implies (7) $\sigma^I \cup \sigma_2' \models \varphi$.
  (6) and (7) imply (8) ${\sigma^I}'' \cup \sigma_2' \models [\vec{x''}/\vec{x}]\varphi \land [\vec{x''}/\vec{x}]\varphi_c$.
  (2) implies (9) $\valuation_1'' \CONTIREACH{[\vec{x''}/\vec{x}]\flow(q_1)}{[\vec{x''}/\vec{x}]\inv(q_1)} {\valuation^I}''$.
  Therefore, from (8) and (9), we have (10) $\sigma_1'' \cup \sigma_2 \models \CONTIREACHPRED{[\vec{x''}/\vec{x}]\flow(q_1)}{[\vec{x''}/\vec{x}]\inv(q_1)}([\vec{x''}/\vec{x}]\varphi \land [\vec{x''}/\vec{x}, \vec{x}/\vec{x'}]\varphi_c)$.
  (Note that the variables in $\vec{x'}$ appear only in $\varphi_c$.)
  $\sigma_1'' \cup \sigma_2 \models [\vec{x''}/\vec{x}]R(\st_1)$ follows from (1); therefore, we have $\sigma_1'' \cup \sigma_2 \models [\vec{x''}/\vec{x}]R(\st_1) \land \CONTIREACHPRED{[\vec{x''}/\vec{x}]\flow(q_1)}{[\vec{x''}/\vec{x}]\inv(q_1)}([\vec{x''}/\vec{x}]\varphi \land [\vec{x''}/\vec{x}, \vec{x}/\vec{x'}]\varphi_c)$.
  This implies $\exists \vec{x''}.\left(
[\vec{x''}/\vec{x}]R(\st) \land \CONTIREACHPRED{[\vec{x''}/\vec{x}]\flow(\st)}{[\vec{x''}/\vec{x}]\inv(\st)}([\vec{x''}/\vec{x}]\varphi \land [\vec{x}/\vec{x'},\vec{x''}/\vec{x}]\cmd) \right)$ as required.
  \end{sloppypar}
  % $\valuation_2 \models \predtranshybrid(R)(\st_2)$ is equivalent to
\end{proof}

\begin{lemma}
  \label{lem:predtranscontProp}
  If $\valuation_1 \models R(\st_1)$ and $\valuation_1 \CONTIREACH{\flow(q_1)}{\inv(q_1)} \valuation_2$, then $\valuation_2 \models \predtranscont(R)(\st_1)$.
\end{lemma}
\begin{proof}
  Almost the same argument as the proof of Lemma~\ref{lem:predtranshybridProp}.
\end{proof}

\subsection{Extension of GPDR}
\label{sec:extendedprocedure}

\begin{figure}[tp]
  \[
    \scriptsize
    \begin{array}{lrcl}
      \INITIALIZE & & \Lra & \PDRState{\emptyset}{\tuple{\rel_0 := \predtranshybrid(\lambda \st. \FALSE) ; \rel_{\rem} := \lambda\st.\TRUE; N := 0}} \\
                  &&&\IF \forall \st \in \states. \models \rel_0(\st) \implies \varphi_P\\
      \VALID & \PDRState{\cetrace}{\abstraction} & \Lra & \RESVALID\\
                  &&&\IF \exists i < N. \forall \st \in \states. \models \rel_{i}(\st) \implies \rel_{i-1}(\st)\\
      \UNFOLD & \PDRState{\cetrace}{\abstraction} & \Lra & \PDRState{\emptyset}{A[\rel_{N+1} := \lambda \st. \TRUE; \rel_{\rem} := \lambda\st.\TRUE; N := N + 1]}\\
                  &&& \IF \forall \st \in \states. \models \rel_{\rem}(\st) \implies \varphi_P\\
      \INDUCTION & \PDRState{\cetrace}{\abstraction} & \Lra & \PDRState{\emptyset}{A[\rel_j := \lambda \st. \rel_j(\st) \land \rel(\st)]_{j=1}^{i+1}}\\
                  &&&\begin{array}[t]{l}\IF % (\fml \lor \psi) \in R_i, \fml \not\in R_{i+1} \\
                       \forall \st \in \states. \models \predtranshybrid(\lambda \st. \rel_i(\st) \land \rel(\st))(\st) \implies \rel(\st)\\
                     \end{array}\\
      % \CANDIDATE & \PDRState{\emptyset}{\abstraction} & \Lra & \PDRState{\tuple{\valuation,\st,N}}{\abstraction}\\
      %             &&& \IF \valuation \models \rel_N(\st) \land \neg \varphi_P \\
      \DECIDE & \PDRState{\tuple{\valuation_2,\st_2,i+1}\cetrace}{\abstraction} & \Lra & \PDRState{\tuple{\valuation_1,\st_1,i}\tuple{\valuation_2,\st_2,i+1}\cetrace}{A}\\
                  &&& \begin{array}[t]{l}
                        \IF
                        \tuple{\st_1,\varphi,\cmd,\st_2} \in \trans
                        \AND
                        \valuation_1,\valuation_2' \models \rel_i(\st_1) \land \CONTIREACHPRED{\flow(\st_1)}{\inv(\st_1)}(\varphi \land \cmd)\\
                        % \mbox{where $\valuation_2$ is obtained by priming every variable in $\valuation'$}
                        % \transrel_{\st,\st'}[\rel_i(\st)] \\
                      \end{array}\\
      \MODEL & \PDRState{\tuple{\valuation,\st_0,0}\cetrace}{A} & \Lra & \RESMODEL \tuple{\valuation,\st_0,0}\cetrace\\
                  % &&& \IF \st = \st_0\\
      \CONFLICT & \PDRState{\tuple{\valuation',\st',i+1}\cetrace}{A} & \Lra & \PDRState{\emptyset}{A[\rel_j := \lambda \st. \rel_j(\st) \land \rel(\st)]_{j=1}^{i+1}}\\ 
                  &&& \IF \models \rel(\st') \implies \neg\valuation' \AND \forall \st \in \states. \models \predtranshybrid(R_i)(\st) \implies \rel(\st)\\
      \INDUCTIONCONT & \PDRState{\cetrace}{\abstraction} & \Lra & \PDRState{\cetrace}{A[\rel_\rem := \lambda \st. \rel_\rem(\st) \land \rel(\st)]}\\
      % &&& \IF (\fml \lor \psi) \in R_i, \fml \not\in R_{\rem} \models \predtranscont(\rel_N \land \fml) \implies \fml\\
                  % &&& \IF (\fml \lor \psi) \in R_i, \fml \not\in R_{\rem} \models \predtranscont(\rel_N \land \fml) \implies \fml\\
                  &&& \begin{array}[t]{l}\IF % (\fml \lor \psi) \in R_i, \fml \not\in R_{i+1} \\
                        \forall \st \in \states. \models \rel_N(\st) \lor \predtranscont(\rel_N)(\st) \implies \rel(\st)\\
                      \end{array}\\
      \CANDIDATECONT & \PDRState{\emptyset}{\abstraction} & \Lra & \PDRState{\tuple{\valuation,\st,\rem}}{\abstraction}\\
                  &&& \IF \valuation \models \rel_\rem(\st) \land \neg \varphi_P \\
      \DECIDECONT & \PDRState{\tuple{\valuation_2,\st,\rem}}{\abstraction} & \Lra & \PDRState{\tuple{\valuation_1,\st,N}\tuple{\valuation_2,\st,\rem}}{A}\\
                  % &&& \IF \ce,\valuation' \models \transrelC[\rel_N(\valuation')] \\
                  &&& \IF \valuation_1,\valuation_2' \models \rel_N(\st) \land \CONTIREACHPRED{\flow(\st)}{\inv(\st)}(\vec{x}=\vec{x'})\\
      \CONFLICTCONT & \PDRState{\tuple{\valuation',\st',\rem}}{A} & \Lra & \PDRState{\emptyset}{A[\rel_\rem := \lambda \st. \rel_\rem(\st) \land \rel(\st)]}\\
                  &&& \IF \rel(\st') \implies \neg\valuation', \AND \models \rel_N(\st') \lor \predtranscont(\rel_N)(\st') \implies \rel(\st')
    \end{array}
  \]
  \caption{The rules for $\HYBRIDPDR$.}
  \label{fig:defHybridPDR}
\end{figure}

\begin{sloppypar}
We present our adaptation of GPDR for hybrid systems, which we call $\HYBRIDPDR$.
Recall that the original GPDR in Section~\ref{sec:vanillapdr} maintains a configuration of the form $\PDRState{\cetrace}{R_0,\dots,R_N; N}$.
$\HYBRIDPDR$ uses a configuration of the form $\PDRState{\cetrace}{R_0,\dots,R_N; R_{\rem}; N}$.
In addition to the information in the original configurations, we add $R_{\rem}$ which we call \emph{remainder frame}.
$R_{\rem}$ overapproximates the states that are reachable from $R_N$ within one flow transition.
\end{sloppypar}
  
Figure~\ref{fig:defHybridPDR} presents the rules for $\HYBRIDPDR$.
The rules from $\INITIALIZE$ to $\CONFLICT$ are the same as Figure~\ref{fig:defVanillaPDR} except that (1) $\INITIALIZE$ and $\UNFOLD$ are adapted so that they set the remainder frame to $\lambda\st.\TRUE$ and (2) $\CANDIDATE$ is dropped.
We explain the newly added rules.
\begin{itemize}
\item $\INDUCTIONCONT$ discovers a fact that holds in $\rel_{\rem}$.
  The side condition $\models \rel_N(\st) \lor \predtranscont(\rel_N)(\st) \implies \rel(\st)$ for any $\st$ guarantees that $\rel(\st)$ is true at the remainder frame; hence $R$ is conjoined to $\rel_{\rem}$.
\item
  $\CANDIDATECONT$ replaces $\CANDIDATE$ in the original procedure.
  It tries to find a candidate from the frame $\rel_{\rem}$.
  The candidate $\tuple{\st,\valuation}$ found here is added to $\cetrace$ in the form $\tuple{\valuation,\st,\rem}$ to denote that $\tuple{\st,\valuation}$ is found at $\rel_{\rem}$.
\item
  $\DECIDECONT$ propagates a counterexample $\tuple{\valuation',\st',\rem}$ found at $\rel_{\rem}$ to the previous frame $\rel_N$.
  This rule computes the candidate to be added to $\cetrace$ by deciding $\valuation \cup \valuation' \models \rel_N(\st) \land \CONTIREACHPRED{\flow(\st)}{\inv(\st)}(\vec{x}=\vec{x'})$, which guarantees that $\valuation$ evolves to $\valuation'$ under the flow dynamics determined by $\flow(\st)$ and $\inv(\st)$.
\item
  $\CONFLICT$ uses $\predtranshybrid$ instead of $\predtrans$ in the original GPDR.
  As in the rule $\CONFLICT$ in GPDR, the frame $R$ in this rule is a generalization of $\neg\valuation'$ which is not backward reachable to $\rel_i$.
  %
  % Intuitively, $R$ separates $\valuation'$ from both the states denoted by $\varphi_0$ and the states that are reachable from $\rel_i$ in a flow transition followed by a jump transition.
\item
  $\CONFLICTCONT$ is the counterpart of $\CONFLICT$ for the frame $\rel_{\rem}$.
  This rule is the same as $\CONFLICT$ except that it uses $\predtranscont$ instead of $\predtranshybrid$; hence, $R$ separates $\valuation'$ from both the states denoted by $\varphi_0$ and the states that are reachable from $\rel_i$ in a flow transition (\emph{not} followed by a jump transition).
  % block the counterexample candidate $\tuple{\valuation',\st',\rem}$.
  %
\end{itemize}

\subsection{Soundness}
\label{sec:extendedprocedureSoundness}

In order to prove the soundness of $\HYBRIDPDR$, we adapt the definition of $\CONSISTENT$ in Definition~\ref{def:originalCon} for $\HYBRIDPDR$.
\begin{definition}
  \label{def:hybridCon}
  Let $\hsts$ be $\tuple{\states,\st_0,\fml_0,\flow,\inv,\trans}$, $\predtranshybrid$ and $\predtranscont$ be the forward predicate transformers determined by $\hsts$, and $\varphi_P$ be the safety condition to be verified.
  A configuration $C$ is said to be \emph{consistent} if it is $\RESVALID$, $\RESMODEL \tuple{\valuation,\st_0,0}\cetrace$, or $\CONSISTENTH(\PDRState{\cetrace}{R_0,\dots,R_N; R_{\rem}; N})$ that satisfies all of the following:
  \begin{itemize}
  \item (Con-A) $R_0(\st_0) = \fml_0$ and $R_0(\st_i) = \FALSE$ if $\st_i \ne \st_0$;
  \item (Con-B-1) $\models R_i(\st) \implies R_{i+1}(\st)$ for any $\st$ and $i < N$;
  \item (Con-B-2) $\models R_N(\st) \implies R_{\rem}(\st)$ for any $\st$;
  \item (Con-C) $\models R_i(\st) \implies \varphi_P$ if $i < N$;
  \item (Con-D-1) $\models \predtranshybrid(R_i)(\st) \implies R_{i+1}(\st)$ for any $i < N$ and $\st$;
  \item (Con-D-2) $\models \predtranscont(R_N)(\st) \implies R_{\rem}(\st)$ for any $\st$;
  \item (Con-E) if $\tuple{\sigma,\st,\rem} \in \cetrace$, then $\sigma \models R_\rem(\st) \land \neg\varphi_P$;
  \item (Con-F-1) if $\tuple{\sigma_1,\st_1,i}, \tuple{\sigma_2,\st_2,i+1} \in \cetrace$ and $i < N$, then $\tuple{\st_1,\varphi,\cmd,\st_2} \in \trans$ and $\sigma_1,\sigma_2' \models R_i(\st_1) \land \varphi \land \varphi_c$; and
  \item (Con-F-2) if $\tuple{\sigma_1,\st_1,N}, \tuple{\sigma_2,\st_2,\rem} \in \cetrace$, then $\tuple{\st_1,\varphi,\cmd,\st_2} \in \trans$ and $\sigma_1,\sigma_2' \models R_i(\st_1) \land \varphi \land \varphi_c$.
  \end{itemize}
\end{definition}

The soundness proof follows the same strategy as that of the original GPDR.

\begin{lemma}
  \label{lem:hybrid-invariant}
  $\CONSISTENTH$ is invariant to any rule application of Figure~\ref{fig:defHybridPDR}.
\end{lemma}

% \todo{Define $\REACHD{}$.}
% \begin{lemma}
%   \begin{itemize}
%   \item
%     \begin{sloppypar}
%     If $\CONSISTENTH(\PDRState{\cetrace}{\rel_0,...,\rel_N; \rel_{\rem}; N})$ holds,
%     then,
%     for all $i \in [0,N]$ and for any $\st_0,\dots,\st_i$ and $\valuation_0,\dots,\valuation_i$,
%     if $\valuation_0 \models \varphi_0$ and
%     $\tuple{\st_0,\valuation_0} \CONTIREACH{\flow(\st_0)}{\inv(\st_0)}\RED{\trans} \tuple{\st_1,\valuation_1} \CONTIREACH{\flow(\st_1)}{\inv(\st_1)}\RED{\trans} \dots \CONTIREACH{\flow(\st_{i-1})}{\inv(\st_{i-1})}\RED{\trans} \tuple{\st_i,\valuation_i}$,
%     then $\valuation_i \models \rel_i(\st_i)$.
%   \end{sloppypar}
% \item
%   \begin{sloppypar}
%   If $\CONSISTENTH(\PDRState{\cetrace}{\rel_0,...,\rel_N; \rel_{\rem}; N})$ holds,
%   then,
%   for any $\st_0,\dots,\st_N,\st_{\rem}$ and $\valuation_0,\dots,\valuation_N,\valuation_{\rem}$,
%     if $\valuation_0 \models \varphi_0$ and
%     $\tuple{\st_0,\valuation_0} \CONTIREACH{\flow(\st_0)}{\inv(\st_0)}\RED{\trans} \tuple{\st_1,\valuation_1} \CONTIREACH{\flow(\st_1)}{\inv(\st_1)}\RED{\trans} \dots \CONTIREACH{\flow(\st_{N-1})}{\inv(\st_{N-1})}\RED{\trans} \tuple{\st_N,\valuation_N}$ and
%     $\valuation_N \CONTIREACH{\flow(\st_N)}{\inv(\st_N)} \valuation_{\rem}$,
%     then $\valuation_{\rem} \models \rel_{\rem}(\st_{\rem})$.
%   \end{sloppypar}
% \end{itemize}
% \end{lemma}

\begin{theorem}
  \label{th:hybrid-soundness}
  If $\HYBRIDPDR$ is started from the rule $\INITIALIZE$ and leads to $\VALID$, then the system is safe.
%  $\PDRState{M}{A}$ satisfies $\VALID$, then the system is safe.	
%  \todo{If $\CONSISTENT(hoge)$ and $hoge$ satisfies $\RESVALID$, then the system is safe.}
  %
  If $\HYBRIDPDR$ is started from the rule $\INITIALIZE$ and leads to $\RESMODEL \tuple{\valuation_0,\st_0,0}\dots\tuple{\valuation_N,\st_N,N}\tuple{\valuation_\rem,\st_\rem,\rem}$, then the system is unsafe.
%  $\PDRState{M}{A}$ satisfies $\VALID$, then the system is safe.	
% \todo{If $\CONSISTENT(hoge)$ and $hoge$ satisfies $\VALID$, then the system is safe.}
\end{theorem}

% \subsection{Discharging verification conditions}

% An execution of $\HYBRIDPDR$ requires several $\DL$ predicates to be discharged in addition to quantifier-free first-order formula that can be discharged by an off-the-shelf SMT solvers.
% %
% Concretely, $\HYBRIDPDR$ needs to discharge $\DL$ formulas in an application of the following rules.
% \begin{itemize}
% \item $\INDUCTION$: $\models \predtranshybrid(\lambda \st. \rel_i(\st) \land \fml)(\st) \implies \fml$
% \item $\DECIDE$ and $\CONFLICT$: 
%   $\valuation,\valuation' \models \rel_i(\st) \land \CONTIREACHPRED{\flow(\st)}{\inv(\st)}(\varphi \land \cmd)$

%   $\tuple{\st,\varphi,\cmd,\st'} \in \trans$

% \end{itemize}

% The rules in Figure~\ref{fig:defHybridPDR} presents a nondeterministic procedure.
% %
% We need to make it a deterministic procedure for implementation.
% %
% We show one of the deterministic procedure, on which our implementation presented in Section~\ref{sec:implementation} is based, derived from the nondeterministic one in Figure~\ref{fig:defHybridPDR}.

\subsection{Operational presentation of $\HYBRIDPDR$}
\label{sec:imperativeProc}

\newcommand\REMOVETRACE{\ensuremath{\mathit{RemoveTrace}}}
\begin{algorithm}[t]
  % \Procedure{$\DETHYBRIDPDR$}{$\hsts$}
  % \SetKw{Kw}{thetext}
  % \SetKwFunction{KwFn}{Fn}
  \SetKwFunction{KwRemoveCti}{$\REMOVETRACE$}
  \DontPrintSemicolon
  \KwIn{Hybrid automaton $\hsts := \tuple{\states,\st_0,\fml_0,\flow,\inv,\trans}$}
  \KwOut{$\RESMODEL(M)$ if $\hsts$ is unsafe; $M$ is a witnessing trace.  $\RESVALID(\rel)$ if $\hsts$ is safe; $\rel$ is an inductive invariant.}
  \tcp{$\INITIALIZE$}
  $N := 0$; $\rel_0 := \lambda\st. (\IF\ \st=\st_0\ \THEN\ \varphi_0\ \ELSE\ \FALSE)$ \;
  $\rel_1 := \TRUE$; $\rel_{\rem} := \TRUE$; $\cetrace := \emptyset$ \;
  % \tcp*{Invariant: $\cetrace = \emptyset$}
  \While{$\TRUE$}{
    \For{$\st \in \states$}{
      \Switch{$\QUERYSATAND(\rel_{\rem}(\st) \land \neg\varphi_P)$}{
        \uCase{$\SAT(\valuation')$}{
          \tcp{$\CANDIDATECONT$}
          $M := \tuple{\st,\valuation,\rem}$\;
          \Switch{\KwRemoveCti{$M$, $\rel_0,\dots,\rel_N,\rel_{\rem}$, N}}{
            \uCase{$\RESVALID(\rel)$}{
              \Return{$\RESVALID(\rel)$}
            }
            \uCase{$\RESCONT(\rel_0,\dots,\rel_N,\rel_{\rem})$}{
              $M := \emptyset$\;
              Update $\rel_0,\dots,\rel_N,\rel_{\rem}$ to the returned frames\;
            }
            \uCase{$\RESMODEL(M)$}{
              \Return{$\RESMODEL(M)$}
            }
          }
        }
        \uCase{$\UNSAT$}{
          \tcp{$\UNFOLD$}
          $M := \emptyset$; $\rel_{N+1}:=\lambda\st.\TRUE$; $\rel_{\rem}:=\lambda\st.\TRUE$; $N := N + 1$\;
        }
      }
      % \textbf{break}\;
    }
  }
  \caption{Definition of $\DETHYBRIDPDR$.}
  \label{proc:dethybridpdr}
\end{algorithm}

\begin{algorithm}[p]
  \scriptsize
  % \Procedure{$\DETHYBRIDPDR$}{$\hsts$}
  % \SetKw{Kw}{thetext}
  % \SetKwFunction{KwFn}{Fn}
  \SetKwFunction{KwRemoveCti}{$\REMOVETRACE$}
  \DontPrintSemicolon
  \KwIn{Hybrid automaton $\hsts := \tuple{\states,\st_0,\fml_0,\flow,\inv,\trans}$; Trace of counterexamples $M$; Frames $\rel_0,\dots,\rel_N,\rel_{\rem}$; Natural number $N$.}
  \KwOut{}
  \While{$M \ne \emptyset$}{
    \uIf{$M = \tuple{\st',\valuation',\rem}M'$}{
      \Switch{$\QUERYSATANDCONT(\rel_{N}(\st') \land \CONTIREACHPRED{\flow(\st')}{\inv(\st')}(\vec{x}=\valuation'(\vec{x}))$}{
        \tcp{$\DECIDECONT$}
        \uCase{$\SAT(\valuation)$}{
          $M := \tuple{\st',\valuation,N}M$
        }
        \tcp{$\CONFLICTCONT$}
        \uCase{$\UNSAT(R)$}{
          $M := \emptyset; \rel_{\rem} := \lambda \st. \rel_{\rem}(\st) \land R(\st)$\;
          \tcp{$\INDUCTIONCONT$}
          \For{$\psi \in \FORMULAS(\rel_N(\st'))$}{
            \Switch{$\QUERYSATANDCONT(\rel_N(\st') \land \CONTIREACHPRED{\flow(\st')}{\inv(\st')}\neg\psi)$}{
              \uCase{$\UNSAT$}{
                $\rel_{\rem}(\st') := \rel_{\rem}(\st') \land \psi$
              }
            }
          }
        }
      }
    }\uElseIf{$M = \tuple{\st',\valuation',0}M'$}{
      \tcp{$\MODEL$}
      \Return{$\RESMODEL(M)$}
    }\ElseIf{$M = \tuple{\st',\valuation',i}M'$ and $0 < i \ne \rem$}{
      \For{$\tuple{\st,\varphi,\varphi_c,\st'} \in \trans$}{
        \Switch{$\QUERYSATANDCONT(\rel_{i-1}(\st) \land \CONTIREACHPRED{\flow(\st)}{\inv(\st)}(\varphi\land\varphi_c\land \vec{x}=\valuation'(\vec{x})))$}{
          \tcp{$\DECIDE$}
          \uCase{$\SAT(\valuation)$}{
            $M := \tuple{\st,\valuation}M$
          }
          \tcp{$\CONFLICT$}
          \uCase{$\UNSAT(R)$}{
            \For{$j \in [1,i+1]$}{
              $\rel_j := \lambda \st. \rel_j(\st) \land R(\st); M := \emptyset$;
            }
          }
          \tcp{$\INDUCTION$}
          \For{$i \in [1,N-1], \psi \in \FORMULAS(\rel_i(\st'))$}{
            \Switch{$\QUERYSATANDCONT(\rel_i(\st') \land \psi \land \CONTIREACHPRED{\flow(\st')}{\inv(\st')}\neg\psi)$}{
              \uCase{$\UNSAT$}{
                $\rel_{j}(\st') := \rel_{j}(\st') \land \psi$ for $j \in [1,i+1]$
              }
            }
          }
        }
        % \textbf{break}\;
      }
    }
  }
  \uIf{There exists $i$ such that $\forall \st. \models \rel_{i+1}(\st) \implies \rel_{i}(\st)$}{
    \tcp{$\VALID$}
    \Return{$\RESVALID(\rel_{i})$}
  }\Else{
    \tcp{Inductive invariant is not reached yet.}
    \Return{$\RESCONT(\rel_0,\dots,\rel_N,\rel_{\rem})$}
  }
  \caption{Definition of $\REMOVETRACE$.}
  \label{proc:removetrace}
\end{algorithm}

The definition of $\HYBRIDPDR$ in Figure~\ref{fig:defHybridPDR} is declarative and nondeterministic.
For the sake of convenience of implementation, we derive an operational procedure from $\HYBRIDPDR$; we call the operational version $\DETHYBRIDPDR$, whose definition is in Algorithm~\ref{proc:dethybridpdr}.
%
% Due to the space limitation, the definition of $\DETHYBRIDPDR$ is deferred to Appendix~\ref{sec:dethybrid}.

\paragraph{Discharging verification conditions.}

\begin{algorithm}[t]
  \scriptsize
  \KwIn{Formula $\delta := \psi \land \CONTIREACHPRED{\ODE}{\varphi_I}(\wedge_{x \in \vars}x = \valuation'(x))$ to be discharged; Number $T > 0$.}
  \KwOut{$\SAT(\valuation_i)$ if $\valuation_i \models \delta$; $\UNSAT(\psi')$ if $\delta$ is unsatisfiable and $\psi'$ is a generalization of $\valuation'$; aborts if satisfiability nor unsatisfiability is proved.}
  \tcp{$\ODE^{-1}$ is the time-inverted ODE of $\ODE$.  Therefore, $p$ is the backward solution of $\ODE$ from $\valuation'$.}
  $\ODE^{-1}$ := the ODE obtained by replacing all the occurrences of the variable $t$ corresponding to the time to $-t$ and negating each time derivative\;
  Solve $\ODE^{-1}$ numerically from the initial point $\valuation_0 := \valuation'$\label{step:solveRevODE}\;
  Let $p := \valuation_0\valuation_1 \dots \valuation_{T-1}$ be the solution obtained at the Step~\ref{step:solveRevODE}\;
  \tcp{$i_1$ is set to $\infty$ if there is no such $i$.}
  $i_1$ := the minimum $i$ such that $\valuation_j \models \varphi_I$ for any $j < i$ and $\valuation_i \models \psi$\;
  \tcp{$i_2$ is set to $\infty$ if there is no such $i$.}
  $i_2$ := the minimum $i$ such that $\valuation_j \models \varphi_I$ for any $j < i$\;
  \uIf{$i_1 < \infty$}{
    \tcp{$\valuation_{i_1}$ witnesses the satisfiability of $\delta$.}
    \Return{$\SAT(\valuation_{i_1})$}
  }\ElseIf{$i_2 < \infty$}{
    \tcp{$\valuation_{i_2}$ is the end point of the $\ODE^{-1}$ with the stay condition $\varphi_I$, but $\valuation_{i_2} \not\models \psi$.  Therefore, $\psi$ is not backward reachable from $\valuation'$ along with $\ODE$.  Currently, the user needs to provide a predicate that can be used for further refinement.}
    Obtain $\psi'$ such that $\models \exists \vec{x_0}.[\vec{x_0}/\vec{x}]\psi \land \CONTIREACHPRED{[\vec{x_0}/\vec{x}]\ODE}{[\vec{x_0}/\vec{x}]\varphi_I}\vec{x_0}=\vec{x} \implies \psi'$ and $\valuation' \not\models \psi'$ from the user\label{step:askuser}\;
    \Return{$\UNSAT(\psi')$}
  }
  \tcp{Cannot conclude neither satisfiability nor unsatisfiabililty.}
  \textbf{abort}
  \caption{Algorithm for discharging $\delta := \psi \land \CONTIREACHPRED{\ODE}{\varphi_I}(\wedge_{x \in \vars}x = \valuation'(x))$.}
  \label{alg:dischargeCont}
\end{algorithm}

\begin{algorithm}
  \KwIn{Formula $\varphi_1 \land \CONTIREACHPRED{\vec{\dot{x}}=\vec{f}(\vec{x})}{\varphi_I}\neg\varphi_2$ to be discharged; Number $r > 0$.}
  \KwOut{$\UNSAT$ or $\OTHERWISE$; if $\UNSAT$ is returned then the input formula is unsatisfiable.}
  \If{$\varphi_1 \land \varphi_2$ is satisfiable}{\label{step:Satcheck}
    \Return{$\OTHERWISE$}
  }
  Let $\DT$ be a fresh symbol\;
  \tcp{Checking $\varphi_1$ is invariant throughout the dynamics determined by $\vec{\dot{x}}=\vec{f}(\vec{x})$ and $\models \varphi_1 \implies \neg\varphi_2$.}
  \If{$r>\DT>0 \land \varphi_1 \land \varphi_I \land \neg[\vec{x}+\vec{f}(\vec{x})\DT/\vec{x}]\varphi_1$ and $\varphi_1 \land \varphi_2$ are unsatisfiable}{\label{step:phi1nonexpansive}
    \Return{$\UNSAT$}
  }
    \tcp{Checking $\neg\varphi_2$ is invariant throughout in the dynamics determined by $\vec{\dot{x}}=\vec{f}(\vec{x})$ and $\models \varphi_1 \implies \neg\varphi_2$.}
  \If{$r>\DT>0 \land \neg\varphi_2 \land \varphi_I \land [\vec{x}+\vec{f}(\vec{x})\DT/\vec{x}]\varphi_2$ and $\varphi_1 \land \varphi_2$ are unsatisfiable}{\label{step:phi2nonexpansive}
    \Return{$\UNSAT$}
  }
  \Return{$\OTHERWISE$}
  \caption{Algorithm for discharging $\varphi_1 \land \CONTIREACHPRED{\vec{\dot{x}}=\vec{f}(\vec{x})}{\varphi_I}\neg\varphi_2$.}
  \label{alg:dischargeUnsat}
\end{algorithm}

An implementation of $\HYBRIDPDR$ needs to discharge verification conditions during verification.
In addition to verification conditions expressed as a satisfiability problem of a first-order predicate, which can be discharged by a standard SMT solver, $\DETHYBRIDPDR$ needs to discharge conditions including a CRP predicate.
Specifically, $\DETHYBRIDPDR$ needs to deal with the following three types of problems.
\todo{Check the explanation.}
\begin{itemize}
\item
  Checking whether $\delta := \psi \land \CONTIREACHPRED{\ODE}{\varphi_I}(\wedge_{x \in \vars}x = \valuation'(x))$ is satisfiable or not for given first-order predicates $\psi$ and $\varphi_I$, an ODE $\ODE$, and a valuation $\valuation'$.
  $\DETHYBRIDPDR$ needs to discharge this type of predicates when it decides which of $\DECIDECONT$ and $\CONFLICTCONT$ should be applied if the top of $\cetrace$ is $\tuple{\valuation',\st',\rem}$.
  We use Algorithm~\ref{alg:dischargeCont} for discharging $\delta$.
  This algorithm searches for a valuation $\valuation_i$ that witnesses the satisfiability of $\delta$ by using a time-inverted simulation of $\ODE$ as follows.
  Concretely, this algorithm numerically simulates $\ODE^{-1}$, the time-inverted ODE of $\ODE$, starting from the point $\set{\vec{x} \mapsto \valuation'(x)}$.
  If it reaches a point $\valuation_i$ that satisfies $\psi$ and if all $\valuation_{i+1} \dots \valuation'$ in the obtained solution satisfy $\varphi_I$, then $\valuation_i$ witnesses the satisfiability of $\delta$.
  If such $\valuation_i$ does not exist but there is $\valuation_i$ such that $\valuation_i \not\models \varphi_I$, then $\psi$ is not backward reachable from $\valuation'$ and hence $\delta$ is unsatisfiable.
  In this case, Algorithm~\ref{alg:dischargeCont} needs to return a predicate that can be used as $\psi'$ in the rule $\CONFLICTCONT$ in Figure~\ref{fig:defHybridPDR}.
  Currently, we assume that the user provides this predicate.
  We expect that we can help this step of discovering $\psi'$ by using techniques for analyzing continuous dynamics (e.g., automated synthesizer of barrier certificates~\cite{DBLP:conf/hybrid/PrajnaJ04} and Flow*~\cite{DBLP:conf/cav/ChenAS13} in combination with Craig interpolant synthesis procedures~\cite{DBLP:conf/aplas/OkudonoNKSKH17,DBLP:conf/cav/AlbarghouthiM13}).
  If neither holds, then we give up the verification by aborting; this may happen if, for example, the value of $T$ is too small.

\item
  Checking whether $\delta' := \psi \land \CONTIREACHPRED{\flow(\st)}{\inv(\st)}(\varphi\land\varphi_c\land \vec{x}=\valuation'(\vec{x}))$ is satisfiable or not.
  $\DETHYBRIDPDR$ needs to solve this problem in the choice between $\DECIDE$ and $\CONFLICT$.
  This query is different from the previous case in that the formula that appears after $\CONTIREACHPRED{\flow(\st)}{\inv(\st)}$ in $\delta'$ is $\varphi\land\varphi_c\land \vec{x}=\valuation'(\vec{x})$, not $\vec{x}=\valuation'(\vec{x})$; therefore, we cannot use numerical simulation to discharge $\delta'$.
  Although it is possible to adapt Algorithm~\ref{alg:dischargeCont} to maintain the sequence of \emph{predicates} $\alpha_0\alpha_1\dots\alpha_{T-1}$ instead of \emph{valuations} so that each $\alpha_i$ becomes the preimage of $\alpha_{i-1}$ by $\ODE$, the preimage computation at each step is prohibitively expensive.
  % \footnote{We present the adapted algorithm in Appendix~\ref{sec:algDischargeDLFormula} as reference.}
  %
  Instead, the current implementation restricts the input system so that there exists at most one $\valuation$ such that $\valuation \models \varphi\land\varphi_c\land \vec{x}=\valuation'(\vec{x})$ for any $\valuation'$; if this is met, then one can safely use Algorithm~\ref{alg:dischargeCont} for discharging $\delta'$.
  Concretely, we allow only $\varphi_c$ that corresponds to the command whose syntax is given by $c ::= \SKIP \mid x := r_1 x + r_2 \mid x := r_1 x - r_2$ where $\SKIP$ is a command that does nothing; $r_1$ and $r_2$ are real constants.
\item
  \begin{sloppypar}
  Checking whether $\varphi_1 \land \CONTIREACHPRED{\ODE}{\varphi_I}\neg\varphi_2$ is unsatisfiable.
  $\DETHYBRIDPDR$ needs to discharge this type of queries when it applies $\INDUCTION$ or $\INDUCTIONCONT$.
  This case is different from the previous case in that (1) $\DETHYBRIDPDR$ may answer $\OTHERWISE$ without aborting the entire verification if unsatisfiability nor satisfiability is proved, and (2) $\DETHYBRIDPDR$ does not need to return a generalization if the given predicate is unsatisfiable.
  We use Algorithm~\ref{alg:dischargeUnsat} to discharge this type of queries.
  This algorithm first checks the satisfiability of $\varphi_1 \land \varphi_2$ in Step~\ref{step:Satcheck}; if it is satisfiable, then so is the entire formula.
  Then, Step~\ref{step:phi1nonexpansive} tries to prove that the entire formula is unsatisfiable by proving (1) $\varphi_1$ is invariant with respect to the dynamics specified by $\ODE$ and $\varphi_I$ and (2) $\varphi_1 \land \varphi_2$ is unsatisfiable.
  In order to prove the former, the algorithm tries the following sufficient condition: For any positive $\DT$ that is smaller than a positive real number $r$, $\models \varphi_i \land \varphi_I \implies [\vec{x}+\vec{f}(\vec{x})\DT/\vec{x}]\varphi_1$, where $\ODE \equiv \vec{\dot{x}} = \vec{f}(\vec{x})$.\footnote{This strategy is inspired by the previous work by one of the authors on nonstandard programming~\cite{DBLP:conf/aplas/NakamuraKSI17,DBLP:conf/popl/SuenagaSH13,DBLP:conf/cav/HasuoS12,DBLP:conf/icalp/SuenagaH11}.}
  Step~\ref{step:phi2nonexpansive} tries the same strategy but tries to prove that $\neg\varphi_2$ is invariant.
  If both attempts fail, then the algorithm returns $\OTHERWISE$.\footnote{If the flow specified by $\ODE$ is a linear or a polynomial, then we can apply the procedure proposed by Liu et al.~\cite{DBLP:conf/emsoft/LiuZZ11}, which is proved to be sound and complete for such a flow.}
  This algorithm could be further enhanced by incorporating automated invariant-synthesis procedures~\cite{DBLP:conf/cav/ColonSS03,DBLP:conf/popl/SankaranarayananSM04,DBLP:conf/emsoft/LiuZZ11}; exploration of this possibilities is left as future work.
  \end{sloppypar}
\end{itemize}

\section{Proof-of-concept implementation}
\label{sec:implementation}

We implemented $\DETHYBRIDPDR$ as a semi-automated verifier.
We note that the current implementation is intended to be a proof of concept; extensive experiments are left as future work.
The snapshot of the source code as of writing can be found at \url{https://github.com/ksuenaga/HybridPDR/tree/master/src}.

The verifier takes a hybrid automaton $\hsts$ specified with $\SPACEEX$ modeling language~\cite{spaceex}, the initial location $\st_0$, the initial condition $\varphi_0$, and the safety condition $\varphi_P$ as input; then, it applies $\DETHYBRIDPDR$ to discover an inductive invariant or a counterexample.
The frontend of the verifier is implemented with OCaml; in the backend, the verifier uses Z3~\cite{DBLP:conf/tacas/MouraB08} and ODEPACK~\cite{odepack} to discharge verification conditions.

As we mentioned in Section~\ref{sec:imperativeProc}, when a candidate counterexample $\tuple{\st',\valuation',i+1}$ turns out to be backward unreachable to $\rel_i$, then our verifier asks for a generalization of $\valuation'$ to the user; concretely, for example in an application of the rule $\CONFLICT$, the user is required to give $\psi$ such that $\models \psi \implies \neg\valuation' \AND \models (\st,\varphi,\varphi_c,\st') \in \delta \land [\vec{x_0}/\vec{x}]R_i(\st) \land \CONTIREACHPRED{[\vec{x_0}/\vec{x}]\flow(\st)}{[\vec{x_0}/\vec{x}]\inv(\st)}[\vec{x_0}/\vec{x}, \vec{x_0}/\vec{x}](\varphi \land \varphi_c) \implies \psi \AND \models \rel_0(\st') \implies \psi$.
Instead of throwing this query at the user in this form, the verifier asks the following question in order to make this process easier for the user for each $(\st,\varphi,\varphi_c,\st') \in \delta$:
\begin{lstlisting}[mathescape,columns=fullflexible,basicstyle=\ttfamily]
Pre:$\rel_i(\st)$; Flow:$\flow(\st)$; Stay:$\inv(\st)$; Guard:$\varphi$; Cmd:$\varphi_c$; CE:$\valuation'$; Init:$\rel_0(\st')$.
\end{lstlisting}
In applying $\CONFLICTCONT$, the verifier omits the fields \texttt{Guard} and \texttt{Cmd}.

We applied the verifier to the hybrid automaton in Figure~\ref{fig:exampleHA} with several initial conditions and the safety condition $\varphi_P := x \le 1$.
We remark that the outputs from the verifier presented here are post-processed for readability.
We explain how verification is conducted in each setting; we write $\ODE$ for the ODE $\dot{x} = -y, \dot{y} = x$.
\todo{Explain by figures.}
\begin{itemize}
\item
  Initial condition $x = 0 \land y = 0$ at location $\st_0$:
  The verifier finds the inductive invariant $\set{\st_0 \mapsto x = 0 \land y = 0, \st_1 \mapsto x = 0 \land y = 0}$ after asking for proofs of unsatisfiability to the user 5 times.
\item
  Initial condition $x \le \frac{1}{2}$ at location $\st_0$:
  The verifier finds a counterexample $\set{x \mapsto 0.490533,  y \mapsto 1.93995}$, from which the system reaches $\set{x \mapsto 2.00100, y \mapsto 0}$.
  The verifier asks 5 questions, one of which is the following:
  \begin{lstlisting}[mathescape,columns=fullflexible,basicstyle=\ttfamily]
    Pre: $(x \le 1 \land y \ge 0) \lor x \le 0.5$; Flow: $\ODE$; Stay: $y \ge 0$;
    Guard: $y \le 0$; Cmd: $\SKIP$; CE: $\set{x \mapsto 0.998516; y \mapsto -1.889365}$;
    Init: $x \le 0.5$.
  \end{lstlisting}
  Notice that the stay condition is $y \ge 0$ and the guard is $y \le 0$; therefore the predicate $y = 0$ holds when a jump transition happens.
  Since the flow specified by $\ODE$ is an anticlockwise circle whose center is $\set{x \mapsto 0, y \mapsto 0}$ with the stay condition $y \ge 0$, the states after the flow dynamics followed by a jump transition is $x \le 0.5 \land y = 0$, which indeed does not intersect with $x = 0.998516 \land y = -1.889365$.
  The verification proceeds by giving $y \ge 0$ as a generalization in this case.

\item
  Initial condition $0 \le x \le \frac{1}{2} \land 0 \le y \le \frac{1}{2}$ at location $\st_0$:
  The verifier finds an inductive invariant
  \[
    R :=
    \left\{
    \begin{array}{l}
      \st_0 \mapsto (y=0 \land 0 \le x \le 0.707107) \lor (0 \le x \le 0.5 \land 0 \le y \le 0.5),\\
      \st_1 \mapsto y=0 \land -0.707107 \le x \le 0
    \end{array}
    \right\}
  \]
  after asking for 8 generalizations to the user.
  This is indeed an inductive invariant.
  Noting $0.707107 \approx \frac{1}{\sqrt{2}}$, we can confirm that (1) the states that are reachable by flow dynamics followed by a jump transition is the set denoted by $R(\st_0)$; the same holds for the transition from $R(\st_1)$; (2) it contains the initial condition $0 \le x \le 0.5 \land 0 \le y \le 0.5$ at location $\st_0$; and (3) it does not intersect with the unsafe region $x > 1$.
  The following is one of the questions that are asked by the verifier:
  \begin{lstlisting}[mathescape,columns=fullflexible,basicstyle=\ttfamily]
    Pre: $(y = 0 \land -0.707107 \le x \le 0) \lor (0 \le x \le 0.5 \land 0 \le y \le 0.5)$;
    Flow: $\ODE$; Stay: $y \le 0$; CE: $\set{x \mapsto 0.998516; y \mapsto -1.889365}$;
    Init: $\FALSE$.
  \end{lstlisting}
  Instead of a precise overapproximation $(x^2 + y^2 = 0.5 \land y \le 0) \lor (0 \le x \le 0.5 \land 0 \le y \le 0.5)$ of the reachable states, we give $(-0.707107 \le y \le 0 \land -0.707107 \le x \le 0.707107) \lor (0 \le x \le 0.5 \land 0 \le y \le 0.5)$, which progresses the verification.
\end{itemize}

\section{Related work}
\label{sec:relatedwork}

Compared to its success in software verification~\cite{DBLP:conf/cav/BirgmeierBW14,DBLP:conf/sat/HoderB12,DBLP:conf/tacas/CimattiGMT14,DBLP:conf/cav/CimattiG12,DBLP:conf/cav/HoderBM11}, IC3/PDR for hybrid systems is less investigated.
HyComp~\cite{DBLP:conf/fmcad/CimattiGMT13,DBLP:conf/tacas/CimattiGMT15} is a model checker that can use several techniques (e.g., IC3, bounded model checking, and $k$-induction) in its backend.
Before verifying a hybrid system, HyComp discretizes its flows so that the verification can be conducted using existing SMT solvers that do not directly deal with continuous-time dynamics.
Compared to HyComp, $\HYBRIDPDR$ does not necessarily require prior discretization for verification.
We are not aware of an IC3/PDR-based model checking algorithm for hybrid systems that does not require prior discretization.

Kindermann et al.~\cite{DBLP:phd/basesearch/Kindermann14,DBLP:conf/formats/KindermannJN12} propose an application of PDR for a timed system---a system that is equipped with \emph{clock variables};
the flow dynamics of a clock variable $c$ is limited to $\dot{c}=1$.
A clock variable may be also reset to a constant in a jump transition.
Kindermann et al. finitely abstract the state space of clock variables by using region abstraction~\cite{DBLP:journals/sttt/Wang04}.
The abstracted system is then verified using the standard PDR procedure.
Later Isenberg et al.~\cite{DBLP:conf/icfem/0002W14} propose a method that abstracts clock variables by using zone abstraction~\cite{DBLP:journals/sttt/BehrmannBLP06}.
They do not deal with a hybrid system whose flow behavior at each location cannot be described by $\dot{c}=1$; the system in Figure~\ref{fig:exampleHA} is out of the scope of their work.

% Cimatti et al.~\cite{DBLP:conf/fmcad/CimattiGMT13} propose a method to synthesize system parameters using PDR.
% %
% Although their presentation of their algorithm is for jump-only transition systems, they apply their method to the benchmarks that come from hybrid system verification.
% %
% They encode a hybrid automaton by discretizing it with a small finite positive real number; the discretization rate is treated as a parameter so that their algorithm synthesizes it automatically.
% %
% Their strategy is in contrast to ours that encode the behavior of a system by extending the logic on which a forward predicate transformer is defined.

Our continuous-reachability predicates (CRP) are inspired by Platzer's $\DL$~\cite{DBLP:journals/ki/Platzer10}.
We may be able to use the theorem prover KeYmaera X for $\DL$ predicates~\cite{DBLP:conf/cade/FultonMQVP15} for our purpose of discharging CRP.

% \paragraph{Verification of hybrid systems.}

% % Abstract interpretation

% % Model checking

% %% HyTech, HyComp

% % dL

% % Nonstandard programming

\section{Conclusion}
\label{sec:conclusion}

We proposed an adaptation of GPDR to hybrid systems.
For this adaptation, we extended the logic on which the forward predicate transformer is defined with the continuous reachability predicates $\CONTIREACHPRED{\ODE}{\varphi_I}\varphi$ inspired by the differential dynamic logic $\DL$.
The extended forward predicate transformer can precisely express the behavior of hybrid systems.
We formalized our procedure $\HYBRIDPDR$ and proved its soundness.
We also implemented it as a semi-automated procedure, which proves the safety of a simple hybrid system in Figure~\ref{fig:exampleHA}.

On top of the current proof-of-concept implementation, we plan to implement a GPDR-based model checker for hybrid systems.
We expect that we need to improve the heuristic used in the application of the rule $\INDUCTION$, where we currently check sufficient conditions of the verification condition.
We are also looking at automating part of the work currently done by human in verification; this is essential when we apply our method to a system with complex continuous-time dynamics.

\noindent\textbf{Acknowledgements.}
We appreciate the comments from the anonymous reviewers, John Toman, and Naoki Kobayashi.  This work is partially supported by JST PRESTO Grant Number JPMJPR15E5,  JSPS KAKENHI Grant Number 19H04084, and JST ERATO MMSD project.

\bibliographystyle{splncs04}
\bibliography{main}

\ifconffinal
\else
\appendix

\iffull
\section{Proof}

\subsection{Soundness of Vanilla PDR}

In the remainder of this section, we fix a DTSTS $\dtsts := \tuple{\states,\st_0,\fml_0,\trans}$; a predicate transformer $\predtrans$ determined by $\dtsts$; and a safety condition $\varphi_P$ to be verified.

\begin{definition}
  For a frame $R$, we write $\sem{R}_\st$ for
  \[
  \displaystyle\exists\st' \in \states. (\st = \st' \land R(\st')).
  \]
\end{definition}

% \begin{definition}[Consistency]
%   We call a configuration \emph{consistent} if it is $\RESVALID$,
%   $\RESMODEL \tuple{\valuation,\st_0,0}\cetrace$, or
%   $\PDRState{\cetrace}{R_0,\dots,R_N; N}$ that satisfies the following
%   conditions:
%   %
%   \begin{description}
%   \item[(Con-A)] $\models \sem{R_0}_\st \iff (\st = \st_0 \land \fml_0)$;
%   \item[(Con-B)] $\models \sem{R_i}_\st \implies \sem{R_{i+1}}_\st$ for any $\st \in \states$ and $i \in \set{0,\dots,N-1}$;
%     % \item $\models R_N(\st) \implies R_{\rem}(\st)$ for any $\st$;
%   \item[(Con-C)] $\models \sem{R_i}_\st \implies \varphi_P$ for any $\st \in \states$ and $i \in \set{0,\dots,N-1}$; and
%   \item[(Con-D)] $\models \predtrans(R_i)(\st) \implies \sem{R_{i+1}}_\st$ for any $i < N$ and $\st$.
%   \end{description}
%   We write $\CONSISTENT(\PDRState{\cetrace}{R_0,\dots,R_N; N})$ if the
%   configuration $\PDRState{\cetrace}{R_0,\dots,R_N; N}$ is consistent.
% \end{definition}

\begin{lemma}
  \label{lem:vanilla-con-init}
  Let $R_0 := \predtrans(\lambda \st. \FALSE)$.  Then, $R_0(\st_0) =
  \fml_0$ and $R_0(\st_i) = \FALSE$ if $\st_i \ne \st_0$.
\end{lemma}
\begin{proof}
  By the definition of $\predtrans$.  The frame $R_0$ is equivalent to
  $\lambda \st'. (\st' = \st_0 \land \varphi_0)$ by
  Definition~\ref{def:predtransOrig}.  Therefore, $R_0(\st_0) =
  \fml_0$ and $R_0(\st_i) = \FALSE$ if $\st_i \ne \st_0$ as required.
\end{proof}

\begin{lemma}
  \label{lem:vanilla-predtrans-strengthen}
  $\models \predtrans(R')(\st) \implies \predtrans(R)(\st)$ for any $R$ and $\st$ if $\models R'(\st) \implies R(\st)$ for any $\st$.
\end{lemma}
\begin{proof}
  Recall the definition of $\predtrans(R)(\st')$:
  \[
  \begin{array}{l}
    (\st' = \st_0 \land \varphi_0) \lor \displaystyle\bigvee_{ (\st,\varphi,\cmd,\st') \in \trans} \exists \vec{x''}.
    \left(
    \begin{array}{ll}
      & [\vec{x''}/\vec{x}]R(\st)\\
      \land & [\vec{x''}/\vec{x}]\varphi \land [\vec{x}/\vec{x'},\vec{x''}/\vec{x}]\cmd
    \end{array}
    \right).
  \end{array}
  \]
  In the above definition, $R(\st)$ appears in the position where the strength of $\predtrans(R)(\st')$ is monotonic with respect to the strength of $R(\st)$.
  Therefore, changing $R$ to a pointwise-stronger frame $R'$ strengthens the entire formula as required.
\end{proof}

\begin{lemma}
  \label{lem:predtrans-implies-something-implies-init-implies-it}
  If $\models \predtrans(R')(\st) \implies R(\st)$, then $\models \st = \st_0 \land \varphi_0 \implies R(\st)$.
\end{lemma}
\begin{proof}
  From the definition of $\predtrans$.
\end{proof}

% \begin{lemma}
%   Consistency is invariant to any rule application of
%   Figure~\ref{fig:defVanillaPDR}.
% \end{lemma}

\begin{pfof}{Lemma~\ref{lem:invariant}}
  Case analysis on the rules in Figure~\ref{fig:defVanillaPDR}.
  \begin{description}
  \item[\INITIALIZE] The condition (Con-A) follows from
    Lemma~\ref{lem:vanilla-con-init}.  The condition (Con-C) follows from the side condition of $\INITIALIZE$.  The other conditions hold vacuously.
  \item[\VALID] The resulting configuration is consistent as required since it is $\RESVALID$.
  \item[\UNFOLD] Let $C := \PDRState{\cetrace}{R_0,\dots,R_N; N}$,
    $\CONSISTENT(C)$, and $\UNFOLD$ is applied to this configuration.
    Let the resulting configuration $C' :=
    \PDRState{\emptyset}{A[\rel_{N+1} := \lambda \st. \TRUE; N := N +
        1]}$.  From the side condition of $\UNFOLD$, we have $\forall
    \st \in \states. \models \rel_N(\st) \implies \varphi_P$.  We show
    $\CONSISTENT(C')$.
    \begin{itemize}
    \item (Con-A) holds since $R_0$ is not unchanged.
    \item For (Con-B), it is sufficient to show $\models \sem{R_N}_\st \implies \sem{R_{N+1}}_\st$ for any $\st \in \states$.  By definition, $\sem{R_{N+1}}_\st$ is equivalent to $\exists \st' \in \states. (\st = \st')$, which is true for any $\st \in \states$.
    \item (Con-C) holds from the above side condition.
    \item For (Con-D), it is sufficient to prove $\models \predtrans(R_N)(\st) \implies \sem{R_{N+1}}_\st$ for any $\st$, which holds since $\sem{R_{N+1}}_\st$ is equivalent to $\exists \st' \in \states. (\st = \st')$.
    \end{itemize}
    (Con-E) and (Con-F) hold vacuously.
  \item[\INDUCTION] Let $C := \PDRState{\cetrace}{R_0,\dots,R_N; N}$, $\CONSISTENT(C)$, and $\INDUCTION$ is applied to this configuration.
    Let $A$ be $R_0,\dots,R_N$.
    Then, the resulting configuration $C'$ is $\PDRState{\cetrace}{A[\rel_j := \lambda \st. \rel_j(\st) \land \rel(\st)]_{j=1}^{i+1}; N}$, where $\models \predtrans(\lambda \st. \rel_i(\st) \land \rel(\st))(\st) \implies \rel(\st)$ for any $\st \in \states$.
    We show $\CONSISTENT(C')$.
    \begin{itemize}
    \item (Con-A) holds since $R_0$ is unchanged.
    \item To prove (Con-B), fix $\st'' \in \states$ and $i' \in \set{0,\dots,N-1}$ arbitrarily.
      We prove $\models \sem{R_{i'}}_{\st''} \implies \sem{R_{i'+1}}_{\st''}$.
      If $i' > 0$, then (Con-B) immediately follows from (Con-B) for $C$.
      Only interesting case is $i' = 0$, in which we must show $\models \sem{R_0}_{\st''} \implies \sem{\lambda \st. R_1(\st) \land R(\st)}_{\st''}$.
      Since $\models \sem{R_0}_{\st''} \implies \sem{R_1}_{\st''}$ follows from the condition (Con-B) for $C$, we show $\models \sem{R_0}_{\st''} \implies R(\st'')$, which follows from Lemma~\ref{lem:predtrans-implies-something-implies-init-implies-it} and (Con-A) for $C$.
      %
      % By the definition of $\predtrans$, the side condition of $\INDUCTION$ (i.e., $\forall \st \in \states. \models \predtrans(\lambda \st. \rel_i(\st) \land R(\st))(\st) \implies R(\st)$) implies that $\models \st = \st_0 \land \varphi_0 \implies R(\st)$ for any $\st$; therefore, we have $\models (\st'' = \st_0 \land R(\st'')) \implies R(\st'')$.
      % %
      % From (Con-A) for $C$, this is equivalent to $\models \sem{R_0}_{\st''} \implies \varphi$, as required.
    \item (Con-C) is trivial since $\INDUCTION$ strengthen each $\sem{R_i}_\st$. 
    \item
      To prove (Con-D), fix $j < N$ and $\st \in \states$ arbitrarily; we show $\models \predtrans(R_j')(\st) \implies \sem{R_{j+1}'}_\st$ for $R_j'$ and $R_{j+1}'$ in $C'$, where $R_k'$ is defined as follows:
      \[
        R_k' :=
        \left\{
        \begin{array}{ll}
          R_0 & (k = 0)\\
          \lambda \st. R_k(\st) \land R(\st) & (1 \le k \le i + 1)\\
          R_k & (i + 1 < k \le N).
        \end{array}
        \right.
      \]
      \begin{itemize}
      \item
        If $j > i + 1$, then (Con-D) for $C'$ follows from (Con-D) for $C$.
      \item If $j = i + 1$, then we are to prove $\models \predtrans(\lambda \st. R_{i+1}(\st) \land R(\st))(\st) \implies \sem{R_{i+2}}_\st$.
        This follows from (Con-D) for $C$ and Lemma~\ref{lem:vanilla-predtrans-strengthen}.
      \item
        If $1 \le j < i + 1$, then we are to prove $\models \predtrans(\lambda \st. R_j(\st) \land R(\st))(\st) \implies \sem{\lambda \st. R_{j+1}(\st) \land R(\st)}_\st$, which is equivalent to
        (1) $\models (\st = \st_0 \land \varphi_0) \implies R_{j+1}(\st) \land R(\st)$ and
        (2) $[\vec{x''}/\vec{x}]R_j(\st) \land [\vec{x''}/\vec{x}]R(\st) \land [\vec{x''}/\vec{x}]\varphi' \land [\vec{x}/\vec{x'},\vec{x''}/\vec{x}]\cmd \implies R_{j+1}(\st) \land R(\st)$ for any $(\st',\varphi',\varphi_c,\st) \in \delta$ and $\vec{x''}$.
        From (Con-D) for $C$, we already have
        $\models \st = \st_0 \land \varphi_0 \implies R_{j+1}(\st)$ and
        $\models [\vec{x''}/\vec{x}]R_j(\st') \land [\vec{x''}/\vec{x}]\varphi' \land [\vec{x}/\vec{x'}, \vec{x''}/\vec{x}]\varphi_c \implies R_{j+1}(\st)$ for any $(\st',\varphi',\varphi_c,\st) \in \delta$ and $\vec{x''}$.
        Therefore, it suffices to show that
        (1') $\models (\st = \st_0 \land \varphi_0) \implies R(\st)$ and
        (2') $\models [\vec{x''}/\vec{x}]R_j(\st) \land [\vec{x''}/\vec{x}]R(\st) \land [\vec{x''}/\vec{x}]\varphi' \land [\vec{x}/\vec{x'},\vec{x''}/\vec{x}]\cmd \implies R(\st)$ for any $(\st',\varphi',\varphi_c,\st) \in \delta$ and $\vec{x''}$.
        The combination of (1') and (2') is equivalent to $\models \predtrans(\lambda \st. R_j(\st) \land R(\st))(\st) \implies R(\st)$, which follows from the side condition of $\INDUCTION$, (Con-B) for $C$, and Lemma~\ref{lem:vanilla-predtrans-strengthen}.
      \item
        If $j = 0$, then we are to prove $\models \predtrans(R_0)(\st) \implies \sem{\lambda \st. R_1(\st) \land R(\st)}_\st$, which is equivalent to $\models \predtrans(R_0)(\st) \implies (R_1(\st) \land R(\st))$.
        By (Con-D) for $C$, it suffices to show $\models \predtrans(R_0)(\st) \implies R(\st)$.
        We show (1) $\models \predtrans(\lambda \st. R_0(\st) \land R(\st))(\st) \implies R(\st)$ and (2) $\models R_0(\st) \implies R_0(\st) \land R(\st)$; then (Con-D) for $C'$ follows from Lemma~\ref{lem:vanilla-predtrans-strengthen}.
        (1) follows from the side condition of $\INDUCTION$, (Con-B) for $C$, and from Lemma~\ref{lem:vanilla-predtrans-strengthen}.
        (2) follows from Lemma~\ref{lem:predtrans-implies-something-implies-init-implies-it}.
      \end{itemize}
    \end{itemize}
    (Con-E) and (Con-F) hold vacuously.
  \item[\CANDIDATE]
    (Con-A), (Con-B), (Con-C), and (Con-D) trivially hold because $A$ is unchanged.  (Con-E) is trivial.  (Con-F) holds vacuously.
  \item[\DECIDE]
    (Con-A), (Con-B), (Con-C), and (Con-D) trivially hold because $A$ is unchanged.  (Con-E) holds because the $N$-th element of $\cetrace$ is unchanged.  (Con-F) is trivial.
  \item[\MODEL]
    $\RESMODEL \tuple{\valuation,\st_0,0}\cetrace$ is consistent by definition.
  \item[\CONFLICT] $C = \PDRState{\tuple{\valuation',\st',i+1}\cetrace}{A}$ and $C' = \PDRState{\emptyset}{A[\rel_j \la \lambda \st. \rel_j(\st) \land \rel(\st)]_{j=1}^{i+1}}$ where $\models \rel(\st') \implies \neg\valuation'$ and $\forall \st \in \states. \models \predtrans(R_i)(\st) \implies \rel(\st)$.  From the condition $\forall \st \in \states. \models \predtrans(R_i)(\st) \implies \rel(\st)$ and Lemma~\ref{lem:vanilla-predtrans-strengthen}, we have $\forall \st \in \states. \models \predtrans(\lambda \st. R_i(\st) \land \rel(\st))(\st) \implies \rel(\st)$, which is the same as the side condition for $\INDUCTION$.  Notice that the rewriting to a configuration by $\CONFLICT$ is the same as that of $\INDUCTION$; therefore, the soundness for this case follows by the same reasoning as the case for $\INDUCTION$.
  \end{description}
\end{pfof}

\begin{pfof}{Theorem~\ref{th:vanillaSoundness}}
  Suppose that an execution of GPDR starts from $\INITIALIZE$ and ends at $\RESVALID$.
  By Lemma~\ref{lem:invariant} and mathematical induction on the length of the execution, the configuration $C$ just before it reaches $\RESVALID$ is consistent.
  Let $C$ be $\PDRState{\cetrace}{\abstraction}$; then, from the side condition of $\VALID$, there exists $i < N$ such that $\forall \st \in \states. \models \rel_{i}(\st) \implies \rel_{i-1}(\st)$.
  For such $R_{i-1}$, we have the following three facts:
  \begin{itemize}
  \item $\models R_0(\st) \implies R_{i-1}(\st)$ for any $\st \in \states$ from (Con-A) and (Con-B);
  \item $\models \predtrans(R_{i-1})(\st) \implies R_{i-1}(\st)$ for any $\st$ from (Con-B) and (Con-D);
  \item $\models R_{i-1}(\st) \implies \varphi_P$ for any $\st \in \states$ from (Con-C).
  \end{itemize}
  Therefore, $R_{i-1}$ is a fixed point that proves unreachability of $\neg\varphi_P$ from $R_0$ via $\predtrans$.
  This leads to the safety of the system since Lemma~\ref{lem:predtransProp} asserts that $\predtrans$ soundly approximates the dynamics of the DTSTS $\dtsts$.

  On the contrary, suppose an execution of GPDR leads to $\RESMODEL \tuple{\valuation_0,\st_0,0}\dots\tuple{\valuation_N,\st_N,N}$ from $\INITIALIZE$.
  Let the configuration one step before the final one be $C$.
  By the same discussion as the previous case, $C$ is consistent.
  We have the following facts about $\tuple{\valuation_0,\st_0,0}\dots\tuple{\valuation_N,\st_N,N}$:
  \begin{itemize}
  \item $\valuation_0 \models \varphi_0$ by (Con-A) and (Con-F);
  \item From (Con-F), for each $\tuple{\valuation^{(1)},\st^{(1)},i}$ and $\tuple{\valuation^{(2)},\st^{(2)},i+1}$, there is a transition from the former to the latter; and
  \item $\valuation_N \models \neg\varphi_P$ by (Con-E).
  \end{itemize}
  Therefore, $\tuple{\st_0,\valuation_0} \RED{\trans} \dots \RED{\trans} \tuple{\st_N,\valuation_N}$ is a valid trace of $\dtsts$, which witnesses that $\dtsts$ is unsafe.
\end{pfof}

% \begin{lemma}
% If $\CONSISTENT(\PDRState{\cetrace}{\rel_0,...,\rel_N; N})$ holds,
% then,
% for all $i \in [0,N]$ and for any $\st_0,\dots,\st_i$ and $\valuation_0,\dots,\valuation_i$,
% if $\valuation_0 \models \varphi_0$ and
% $\tuple{\st_0,\valuation_0} \RED{\trans} \dots \RED{\trans} \tuple{\st_i,\valuation_i}$,
% then $\valuation_i \models \rel_i(\st_i)$.
% \end{lemma}

% \begin{theorem}
%   If the vanilla PDR procedure is started from the rule $\INITIALIZE$ and leads to $\VALID$, then the system is safe.
% %  $\PDRState{M}{A}$ satisfies $\VALID$, then the system is safe.	
% % \todo{If $\CONSISTENT(hoge)$ and $hoge$ satisfies $\RESVALID$, then the system is safe.}
% \end{theorem}

% \begin{theorem}
%   If the vanilla PDR procedure is started from the rule $\INITIALIZE$ and leads to $\RESMODEL \tuple{\valuation_0,\st_0,0}\dots\tuple{\valuation_N,\st_N,N}$, then the system is unsafe.
% %  $\PDRState{M}{A}$ satisfies $\VALID$, then the system is safe.	
% % \todo{If $\CONSISTENT(hoge)$ and $hoge$ satisfies $\VALID$, then the system is safe.}
% \end{theorem}

\subsection{Soundness of $\HYBRIDPDR$}

In the remainder of this section, set $\hsts$ to $\tuple{\states,\st_0,\fml_0,\flow,\inv,\trans}$, $\predtranshybrid$ and $\predtranscont$ to the forward predicate transformers determined by $\hsts$, and $\varphi_P$ to a safety condition to be verified.

% \begin{definition}
%   We write $\CONSISTENTH(\PDRState{\cetrace}{R_0,\dots,R_N; R_{\rem}; N})$ if the following conditions are met:
%   % 
%   \begin{itemize}
%   \item (Con-A) $R_0(\st_0) = \fml_0$ and $R_0(\st_i) = \FALSE$ if $\st_i \ne \st_0$;
%   \item (Con-B) $\models R_i(\st) \implies R_{i+1}(\st)$ for any $\st$ and $i < N$;
%   \item (Con-C) $\models R_N(\st) \implies R_{\rem}(\st)$ for any $\st$.
%   \item $\models R_i(\st) \implies \varphi_P$ if $i < N$;
%   \item $\models \predtranshybrid(R_i)(\st) \implies R_{i+1}(\st)$ for any $i < N$ and $\st$; and
%   \item $\models \predtranscont(R_N)(\st) \implies R_{\rem}(\st)$ for any $i < N$ and $\st$.
%   \end{itemize}
% \end{definition}

% \begin{lemma}
%   \label{lem:hybrid-con-init}
%   Let $R_0 := \predtranshybrid(\lambda \st. \FALSE)$.  Then, $R_0(\st_0) =
%   \fml_0$ and $R_0(\st_i) = \FALSE$ if $\st_i \ne \st_0$.
% \end{lemma}
% \begin{proof}
%   By the definition of $\predtranshybrid$.  The proof is almost the same as that of Lemma~\ref{lem:vanilla-con-init}.
% \end{proof}

% \begin{lemma}
%   \label{lem:vanilla-predtrans-strengthen}
%   $\models \predtranshybrid(R')(\st) \implies \predtranshybrid(R)(\st)$ for any $R$ and $\st$ if $\models R'(\st) \implies R(\st)$ for any $\st$.
% \end{lemma}

\begin{lemma}
  \label{lem:hybrid-con-init}
  Let $R_0 := \predtranshybrid(\lambda \st. \FALSE)$.  Then,
  $R_0(\st_0) = \fml_0$ and $R_0(\st_i) = \FALSE$ if
  $\st_i \ne \st_0$.
\end{lemma}
\begin{proof}
  By the definition of $\predtranshybrid$.  The proof is the same argument as that of Lemma~\ref{lem:vanilla-con-init}.
\end{proof}

\begin{lemma}
  \label{lem:hybrid-predtrans-strengthen}
  $\models \predtranshybrid(R')(\st) \implies \predtranshybrid(R)(\st)$ for any $R$ and $\st$ if $\models R'(\st) \implies R(\st)$ for any $\st$.
\end{lemma}
\begin{proof}
  By the definition of $\predtranshybrid$.  The proof is the same argument as that of Lemma~\ref{lem:vanilla-predtrans-strengthen}.
\end{proof}

\begin{lemma}
  \label{lem:hybrid-predtrans-implies-something-implies-init-implies-it}
  If $\models \predtranshybrid(R')(\st) \implies R(\st)$, then $\models \st = \st_0 \land \varphi_0 \implies R(\st)$.
\end{lemma}
\begin{proof}
  From the definition of $\predtranshybrid$.
\end{proof}

\begin{pfof}{Lemma~\ref{lem:hybrid-invariant}}
  Case analysis on the rules in Figure~\ref{fig:defHybridPDR}.
  We omit the cases whose proof is almost the same as that of Lemma~\ref{lem:invariant}.
  \begin{description}
  % \item[\INITIALIZE]
  %   Lemma~\ref{lem:vanilla-con-init}
  % \item[\VALID]
  \item[\UNFOLD] Let
    $C := \PDRState{\cetrace}{R_0,\dots,R_N,R_{\rem}; N}$,
    $\CONSISTENTH(C)$, and $\UNFOLD$ is applied to this configuration.
    Let the resulting configuration
    $C' := \PDRState{\emptyset}{A[\rel_{N+1} := \lambda \st. \TRUE,
      \rel_{\rem} := \lambda\st.\TRUE; N := N + 1]}$.  From the side
    condition of $\UNFOLD$, we have
    $\forall \st \in \states. \models \rel_{\rem}(\st) \implies
    \varphi_P$.  We show $\CONSISTENTH(C')$.
    \begin{itemize}
    \item For (Con-B), it is sufficient to show
      $\models \sem{R_N}_\st \implies \sem{R_{N+1}}_\st$ for any
      $\st \in \states$.  By definition, $\sem{R_{N+1}}_\st$ is
      equivalent to $\exists \st' \in \states. (\st = \st')$, which is
      true for any $\st \in \states$.
    \item (Con-C) holds from the above side condition and (Con-B-1)
      and (Con-B-2).
    \item For (Con-D-1) and (Con-D-2) are similar to the proof of Lemma~\ref{lem:invariant}.
      % it is sufficient to prove
      % $\models \predtranshybrid(R_N)(\st) \implies
      % \sem{R_{N+1}}_\st$
      % for any $\st$, which holds since $\sem{R_{N+1}}_\st$ is
      % equivalent to $\exists \st' \in \states. (\st = \st')$.
    \end{itemize}
    (Con-E), (Con-F-1), and (Con-F-2) hold vacuously.
  % \item[\DECIDE] (Con-A), (Con-B-1), (Con-B-2), (Con-C), (Con-D-1),
  %   and (Con-D-2) trivially hold because $A$ is unchanged.  (Con-E)
  %   holds because the $\rem$ element of $\cetrace$ is unchanged.
  %   (Con-F) is trivial.
  % \item[\MODEL]
  \item[\CONFLICT] Same as the proof of Lemma~\ref{lem:invariant}
    wherein we use Lemma~\ref{lem:hybrid-predtrans-strengthen} instead
    of Lemma~\ref{lem:vanilla-predtrans-strengthen}.
  % \item[\CANDIDATECONT]
  % \item[\DECIDECONT]
   
  \item[\INDUCTION]
    Let $C := \PDRState{\cetrace}{R_0,\dots,R_N,R_{\rem}; N}$, $\CONSISTENTH(C)$, and $\INDUCTION$ is applied to this configuration.
    Let $A$ be $R_0,\dots,R_N,R_{\rem}$.
    Then, the resulting configuration $C'$ is $\PDRState{\cetrace}{A[\rel_j := \lambda \st. \rel_j(\st) \land \rel(\st)]_{j=1}^{i+1}; N}$, where $\models \predtranshybrid(\lambda \st. \rel_i(\st) \land \rel(\st))(\st) \implies \rel(\st)$ for any $\st \in \states$.
    We show $\CONSISTENTH(C')$.
    \begin{itemize}
    \item (Con-A) holds since $R_0$ is unchanged.
    \item To prove (Con-B-1), fix $\st'' \in \states$ and $i' \in \set{0,\dots,N-1}$ arbitrarily.
      We prove $\models \sem{R_{i'}}_{\st''} \implies \sem{R_{i'+1}}_{\st''}$.
      If $i' > 0$, then (Con-B) immediately follows from (Con-B) for $C$.
      Only interesting case is $i' = 0$, in which we must show $\models \sem{R_0}_{\st''} \implies \sem{\lambda \st. R_1(\st) \land R(\st)}_{\st''}$.
      Since $\models \sem{R_0}_{\st''} \implies \sem{R_1}_{\st''}$ follows from the condition (Con-B) for $C$, we show $\models \sem{R_0}_{\st''} \implies R(\st'')$, which follows from Lemma~\ref{lem:hybrid-predtrans-implies-something-implies-init-implies-it} and (Con-A) for $C$.
      %
      % By the definition of $\predtrans$, the side condition of $\INDUCTION$ (i.e., $\forall \st \in \states. \models \predtrans(\lambda \st. \rel_i(\st) \land R(\st))(\st) \implies R(\st)$) implies that $\models \st = \st_0 \land \varphi_0 \implies R(\st)$ for any $\st$; therefore, we have $\models (\st'' = \st_0 \land R(\st'')) \implies R(\st'')$.
      % %
      % From (Con-A) for $C$, this is equivalent to $\models \sem{R_0}_{\st''} \implies \varphi$, as required.
    \item (Con-B-2) trivially holds since $\INDUCTION$ does not change $R_{\rem}$.
    \item (Con-C) is trivial since $\INDUCTION$ strengthen each $\sem{R_i}_\st$. 
    \item
      To prove (Con-D-1), fix $j < N$ and $\st \in \states$ arbitrarily; we show $\models \predtranshybrid(R_j')(\st) \implies \sem{R_{j+1}'}_\st$ for $R_j'$ and $R_{j+1}'$ in $C'$, where $R_k'$ is defined as follows:
      \[
        R_k' :=
        \left\{
        \begin{array}{ll}
          R_0 & (k = 0)\\
          \lambda \st. R_k(\st) \land R(\st) & (1 \le k \le i + 1)\\
          R_k & (i + 1 < k \le N).
        \end{array}
        \right.
      \]
      \begin{itemize}
      \item
        If $j > i + 1$, then (Con-D-1) for $C'$ follows from (Con-D-1) for $C$.
      \item If $j = i + 1$, then we are to prove $\models \predtranshybrid(\lambda \st. R_{i+1}(\st) \land R(\st))(\st) \implies \sem{R_{i+2}}_\st$.
        This follows from (Con-D-1) for $C$ and Lemma~\ref{lem:hybrid-predtrans-strengthen}.
      \item
        If $1 \le j < i + 1$, then we are to prove $\models \predtranshybrid(\lambda \st. R_j(\st) \land R(\st))(\st) \implies \sem{\lambda \st. R_{j+1}(\st) \land R(\st)}_\st$, which is equivalent to
        (1) $\models (\st = \st_0 \land \varphi_0) \implies R_{j+1}(\st) \land R(\st)$ and
        (2) $[\vec{x''}/\vec{x}]R(\st) \land \CONTIREACHPRED{[\vec{x''}/\vec{x}]\flow(\st)}{[\vec{x''}/\vec{x}]\inv(\st)}([\vec{x''}/\vec{x}]\varphi \land [\vec{x}/\vec{x'},\vec{x''}/\vec{x}]\cmd \implies \sem{\lambda \st. R_{j+1}(\st) \land R(\st)}_\st$
        for any $(\st',\varphi',\varphi_c,\st) \in \delta$ and $\vec{x''}$.
        From (Con-D-1) for $C$, we already have
        $\models \st = \st_0 \land \varphi_0 \implies R_{j+1}(\st)$ and
        $[\vec{x''}/\vec{x}]R(\st) \land \CONTIREACHPRED{[\vec{x''}/\vec{x}]\flow(\st)}{[\vec{x''}/\vec{x}]\inv(\st)}([\vec{x''}/\vec{x}]\varphi \land [\vec{x}/\vec{x'},\vec{x''}/\vec{x}]\cmd \implies \sem{\lambda \st. R_{j+1}(\st)}_\st$ for any $(\st',\varphi',\varphi_c,\st) \in \delta$ and $\vec{x''}$.
        Therefore, it suffices to show that
        (1') $\models (\st = \st_0 \land \varphi_0) \implies R(\st)$ and
        (2') $[\vec{x''}/\vec{x}]R(\st) \land \CONTIREACHPRED{[\vec{x''}/\vec{x}]\flow(\st)}{[\vec{x''}/\vec{x}]\inv(\st)}([\vec{x''}/\vec{x}]\varphi \land [\vec{x}/\vec{x'},\vec{x''}/\vec{x}]\cmd \implies R(\st)$ for any $(\st',\varphi',\varphi_c,\st) \in \delta$ and $\vec{x''}$.
        The combination of (1') and (2') is equivalent to $\models \predtrans(\lambda \st. R_j(\st) \land R(\st))(\st) \implies R(\st)$, which follows from the side condition of $\INDUCTION$, (Con-B-1) for $C$, and Lemma~\ref{lem:hybrid-predtrans-strengthen}.
      \item
        If $j = 0$, then we are to prove $\models \predtranshybrid(R_0)(\st) \implies \sem{\lambda \st. R_1(\st) \land R(\st)}_\st$, which is equivalent to $\models \predtranshybrid(R_0)(\st) \implies (R_1(\st) \land R(\st))$.
        By (Con-B-1) for $C$, it suffices to show $\models \predtranshybrid(R_0)(\st) \implies R(\st)$.
        We show (1) $\models \predtranshybrid(\lambda \st. R_0(\st) \land R(\st))(\st) \implies R(\st)$ and (2) $\models R_0(\st) \implies R_0(\st) \land R(\st)$; then (Con-D-1) for $C'$ follows from Lemma~\ref{lem:hybrid-predtrans-strengthen}.
        (1) follows from the side condition of $\INDUCTION$, (Con-B-1) for $C$, and from Lemma~\ref{lem:hybrid-predtrans-strengthen}.
        (2) follows from Lemma~\ref{lem:hybrid-predtrans-implies-something-implies-init-implies-it}.
      \end{itemize}
    \end{itemize}
    (Con-E) and (Con-F) hold vacuously.
  \item[\INDUCTIONCONT]
    Let $C := \PDRState{\cetrace}{R_0,\dots,R_N,R_{\rem}; N}$, $\CONSISTENTH(C)$, and $\INDUCTIONCONT$ is applied to this configuration.
    Let $A$ be $R_0,\dots,R_N,R_{\rem}$.
    Then, the resulting configuration $C'$ is $\PDRState{\cetrace}{A[\rel_\rem := \lambda \st. \rel_\rem(\st) \land \rel(\st)]}$, where $\models \rel_N(\st) \lor \predtranscont(\rel_N)(\st) \implies \rel(\st)$ for any $\st \in \states$.
    We show $\CONSISTENTH(C')$.
    \begin{itemize}
    \item (Con-A) holds since $R_0$ is unchanged.
    \item (Con-B-1), (Con-C), and (Con-D-1) trivially hold since $\INDUCTIONCONT$ changes only $R_{\rem}$.
    \item To prove (Con-B-2), fix $\st'' \in \states$ arbitrarily.
      We are to prove $\models R_N(\st'') \implies R_{\rem}(\st'') \land R(\st'')$.
      From (Con-B-2), it suffices to show $\models R_N(\st'') \implies R(\st'')$, which follows from the side condition of $\INDUCTIONCONT$.
    % \item (Con-C) is trivial since $R_0,\dots,R_{N-1}$ are not changed by $\INDUCTIONCONT$.
    \item
      To prove (Con-D-2), fix $\st \in \states$ arbitrarily.
      We are to show $\models \predtranscont(R_N)(\st) \implies R_{\rem}(\st) \land R(\st)$, which follows from (Con-D-2) for $C$ and the side condition for $\INDUCTIONCONT$.
    \end{itemize}
    (Con-E) and (Con-F) hold vacuously.
  \item[\CONFLICTCONT]
    The argument for this case is almost the same as that of $\CONFLICT$.
  \end{description}
\end{pfof}

\begin{pfof}{Theorem~\ref{th:hybrid-soundness}}
  Suppose that an execution of GPDR starts from $\INITIALIZE$ and ends at $\RESVALID$.
  By Lemma~\ref{lem:hybrid-invariant} and mathematical induction on the length of the execution, the configuration $C$ just before it reaches $\RESVALID$ is consistent.
  Let $C$ be $\PDRState{\cetrace}{\abstraction}$; then, from the side condition of $\VALID$, there exists $i < N$ such that $\forall \st \in \states. \models \rel_{i}(\st) \implies \rel_{i-1}(\st)$.
  By the same argument as the proof of Theorem~\ref{th:vanillaSoundness}, we can show that $R_{i-1}$ is a fixed point that proves unreachability of $\neg\varphi_P$ from $R_0$ via $\predtranshybrid$.
  This leads to the safety of the system since Lemma~\ref{lem:predtranshybridProp} asserts that $\predtranshybrid$ soundly approximates the dynamics of the HA $\hsts$.

  On the contrary, suppose an execution of GPDR leads to
  \[
    \RESMODEL \tuple{\valuation_0,\st_0,0}\dots\tuple{\valuation_N,\st_N,N}\tuple{\valuation_{\rem},\st_N,\rem}
  \]
  from $\INITIALIZE$.
  Let the configuration one step before the final one be $C$.
  By the same discussion as the previous case, $C$ is consistent.
  We have the following facts about $\tuple{\valuation_0,\st_0,0}\dots\tuple{\valuation_N,\st_N,N}\tuple{\valuation_{\rem},\st_N,\rem}$:
  \begin{itemize}
  \item $\valuation_0 \models \varphi_0$ by (Con-A) and (Con-F);
  \item From (Con-F-1), for each $\tuple{\valuation^{(1)},\st^{(1)},i}$ and $\tuple{\valuation^{(2)},\st^{(2)},i+1}$, there is a transition from the former to the latter;
  \item From (Con-F-2), for each $\tuple{\valuation^{(1)},\st^{(1)},N}$ and $\tuple{\valuation^{(2)},\st^{(2)},\rem}$, there is a continuous transition from the former to the latter; and
  \item $\valuation_N \models \neg\varphi_P$ by (Con-E).
  \end{itemize}
  Therefore, the run that consists of $\tuple{\st_0,\valuation_0} \dots \tuple{\st_N,\valuation_N} \tuple{\st_\rem,\valuation_\rem}$ is a valid trace of $\hsts$, which witnesses that $\hsts$ is unsafe.
\end{pfof}

\fi

\fi

\end{document}